\documentclass[9pt,a4paper,web]{article}
  \def\appendixname{}

\usepackage{graphicx}
\usepackage{mathtools,amsmath,amsthm}
\usepackage{mathalfa}
\usepackage[bitstream-charter]{mathdesign}
\usepackage{color}
\usepackage{etoolbox}
\usepackage{url}
\usepackage{wrapfig}

\usepackage{algorithmic}
\usepackage{graphicx}
\usepackage[normalem]{ulem}
\usepackage{textcomp}

\usepackage{enumitem}
\usepackage{color}
\usepackage{float}

\newtheorem{theorem}{Theorem}
\newtheorem{corollary}{Corollary}
\newtheorem{definition}{Definition}
\newtheorem{lemma}{Lemma}

\newtheorem{proposition}{Proposition}

\newtheorem{remma}{Remark}
\newtheorem{property}{Property}

\renewenvironment{proof}{\smallskip \noindent{\textbf{Proof:}}}{\hfill \hspace*{1pt}\hfill $\blacksquare$}
\renewenvironment{proofname}[1]{\smallskip \noindent {\textbf{#1}}}{\hfill \hspace*{1pt}\hfill $\blacksquare$}
\newenvironment{remark}{\begin{remma}\rm }{\hfill \hspace*{1pt} \hfill $\square$\end{remma}}

\DeclareMathOperator{\argmin}{argmin}
\DeclareMathOperator{\sat}{sat}

\newcommand{\smallmat}[1]{\left[ \begin{smallmatrix}#1
	\end{smallmatrix} \right]}

\newcommand{\atan}{\ensuremath{\mathrm{atan}}}
\newcommand{\dom}{\ensuremath{\mathrm{dom}}\,}
\newcommand{\sgn}{\ensuremath{\mathrm{sign}}}
\newcommand{\nulls}{\ensuremath{\mathrm{null}}}

\newcommand{\gph}{\ensuremath{\mathrm{gph}\ }}
\newcommand{\A}{\ensuremath{\mathcal{A}}}

\newcommand{\E}{\ensuremath{\mathcal{E}}}
\newcommand{\ext}{\ensuremath{\mathrm{ext}}}
\newcommand{\G}{\ensuremath{\mathcal{G}}}

\newcommand{\K}{\ensuremath{\mathcal{K}}}
\newcommand{\KL}{\ensuremath{\mathcal{KL}}}
\newcommand{\V}{\ensuremath{\mathcal{V}}}
\newcommand{\real}{\ensuremath{{\mathbb R}}}
\newcommand{\integer}{\ensuremath{{\mathbb Z}}}
\newcommand{\domain}{\ensuremath{\mathrm{dom\hspace{0.5mm}}}}

\newcommand{\cl}{\ensuremath{\mathrm{cl}}}
\newcommand{\sign}{\ensuremath{\mathrm{sign}}}
\newcommand{\rank}{\ensuremath{\mathrm{rank}}}

\newcommand{\uni}{\ensuremath{\mathrm{uni}}}

\graphicspath{{figures/}}
\begin{document}

\title{\Huge{Hybrid coupling rules for leaderless heterogeneous oscillators: uniform global asymptotic and finite-time synchronization} \thanks{Work supported by the ANR under grant HANDY  ANR-18-CE40-0010. Corresponding author S. Mariano. A preliminary version of this work was presented at the 21th IFAC World Congress, Berlin, Germany (\cite{IFACKura}). }}
\author{S. Mariano\thanks{S. Mariano is with the Department of Electrical and Electronic Engineering, University of Melbourne, Australia (e-mail:simone.mariano@unimelb.edu.au).},
R. Bertollo \thanks{R. Bertollo is with Department of Mechanical Engineering, TU Eindhoven, The Netherlands (e-mail:r.bertollo@tue.nl)},
R. Postoyan\thanks{R. Postoyan is with Universit\'e de
	Lorraine, CNRS, CRAN, F-54000 Nancy, France (e-mail:romain.postoyan@univ-lorraine.fr).},
L. Zaccarian\thanks{L. Zaccarian is with LAAS-CNRS, Universit\'e de Toulouse, CNRS, Toulouse, France
	and the Department of Industrial Engineering, University of Trento, Trento, Italy (e-mail: zaccarian@laas.fr).}}	

 \date{}
\maketitle

\begin{abstract}
We investigate the engineering scenario where the objective is to synchronize heterogeneous oscillators in a distributed fashion. The internal dynamics of each oscillator are general enough to capture their time-varying natural frequency as well as physical couplings and unknown bounded terms. A communication layer is set in place to allow the oscillators to exchange synchronizing coupling actions through a tree-like leaderless network. In particular, we present a class of hybrid coupling rules depending only on local information to ensure uniform global practical or asymptotic synchronization, which is impossible to obtain by using the Kuramoto model customarily used in the literature. We further show that the synchronization set can be made uniformly globally prescribed finite-time stable by selecting the coupling function to be discontinuous at the origin. Novel mathematical tools on non-pathological functions and set-valued Lie derivatives are developed to carry out the stability analysis. The effectiveness of the approach is illustrated in simulations where we apply our synchronizing hybrid coupling rules to models of power grids previously used in the literature.
\end{abstract}

\section{Introduction}
\label{sec:introduction-introduction}

The Kuramoto model (\cite{kuramoto1975self}) is used in various research fields to describe and analyze the dynamics of a broad family of systems with oscillatory behavior (\cite{AcebronKuramoto})
 including neuroscience (\cite{aoki2015self,TassDeepBrain,CuminUnsworthNeurons}), chemistry (\cite{FlameDynamics}), power networks (\cite{dorfler2012synchronization})
and natural sciences (\cite{LeonardAnimalGroups}), to cite a few (see also (\cite{strogatz2003sync})). 
The many application areas where Kuramoto dynamics emerged from physical considerations motivated a detailed analysis of the synchronization properties of the model, first for the all-to-all connection case (\cite{Aey1}), as originally described by Kuramoto, then for a
general interconnection layout (\cite{jad1}), with a focus on the derivation of the least conservative lower
bound for a stabilizing coupling gain (\cite{JafarpourBullo,chopra,BulloCriticalCoupling}).

Given its simple and accurate description of natural synchronization phenomena, the Kuramoto
model has also inspired the design of distributed communication protocols in engineering applications where the coupling function among different agents
can be arbitrarily assigned to achieve synchronization, as in the bio-inspired synchronization of moving particles in (\cite{sepulchre2007stabilization}), the synchronized acquisition of oceanographic data from Autonomous Underwater Vehicles (\cite{baldoni2007adaptive}), in clock synchronization (\cite{kiss2018synchronization}), in mobile sensors networks modeled as particles with coupled oscillator dynamics (\cite{paley2007oscillator}), in monotone coupled oscillators (\cite{mauroy2012contraction}) or in other engineering applications surveyed in (\cite{dorfler2014synchronization}).

While the sinusoidal coupling of Kuramoto models provides a powerful tool to obtain synchronization in coupled networks of oscillators, it also introduces some undesirable properties for engineering applications.
For example, when the network comprises oscillators with the same natural frequency,
it is now well-known that a system of Kuramoto oscillators admits, in addition to stable equilibria coinciding with the synchronization set, equilibria that are unstable (see, e.g., (\cite{Strogatz00,sepulchre2007stabilization})). The downside of this result is that the closer a solution is initialized to an unstable equilibrium, the longer it will take for phase synchronization to arise: we talk of \emph{non-uniform} convergence (\cite{sepulchre2007stabilization}). Although non-uniform synchronization may naturally characterize certain physical (\cite{OudMetronomes}) and biological systems, in general, it is not a desirable property for engineering applications. Indeed, the lack of uniformity may induce arbitrarily slow convergence to the attractor set and poor robustness properties (\cite{miller1997maneuvering}). Secondly, it may occur  in the Kuramoto model that the angular phase mismatch between adjacent oscillators remains constant and different from zero indefinitely: in this case we talk of \emph{phase locking} (\cite{Aey1}), which hampers the capability to reach asymptotic collective synchronization.
Thirdly, in critical applications, finite-time stability, instead of only asymptotic synchronization, may be a mandatory requirement (\cite{Polyakov11}).  

In this work, we investigate the engineering scenario where the goal is to synthesize local coupling rules to synchronize a set of heterogeneous oscillators. We assume the model of the oscillators to be general enough to capture not only their (time-varying) natural frequency but also physical coupling actions and other unknown bounded terms, thus being able to represent, among many possibilities, networks of Kuramoto oscillators with heterogeneous time-varying natural frequencies. Furthermore, without loss of generality, we introduce suitable resets of the oscillators' phase coordinates, so that they are unwrapped to evolve in a compact set, which includes $[-\pi,\pi]$ consistently with their angular nature.  Consequently, we define hybrid $2 \pi$-unwinding mechanisms to ensure the forward completeness of the oscillating solutions.  

To achieve uniform global phase synchronization, thereby overcoming the limitations of Kuramoto models, we equip the oscillators with a leaderless tree-like communication network to locally exchange coupling actions based on local information. This approach has been already exploited in the context of DC microgrids as in, e.g., (\cite{cucuzzella2018robust}), or (\cite{giraldo2019synchronisation}), for a network of Kuramoto oscillators equipped with a leader. The selection of a tree-like graph, which can always be derived in a distributed way by using the algorithms surveyed in (\cite{pandurangan2018distributed}), is also not new while addressing a problem of distributed cooperative control: see (\cite{HTsync}) in the context of hybrid dynamical systems, or (\cite{bai2011cooperative}) and (\cite{alagoz2012user}) for continuous-time networked systems and power grids, respectively. To define the coupling actions, we present novel hybrid coupling rules for which a Lyapunov-based analysis ensures uniform global (practical or asymptotic) phase synchronization. This result overcomes both the lack of uniform convergence and the phase-locking issues characterizing the Kuramoto model (\cite{sepulchre2007stabilization}). Interestingly, we can design the coupling rules in such a way that the network of oscillators behaves like the original Kuramoto models when the oscillators are near phase synchronization. Furthermore, due to the mild properties that we require for our hybrid coupling function, discontinuous selections are allowed, like in (\cite{coraggio2020distributed}). When the discontinuity is at the origin, we prove finite-time stability properties. In particular, exact synchronization can be reached in a prescribed finite-time  (\cite{Krstic17}), and convergence is thus independent of the initial conditions. Compared to the related works in (\cite{mauroy2012contraction}) and (\cite{Wu18}), the finite-time stability property we ensure is global and the convergence time can be arbitrarily prescribed, respectively. We resort for this purpose to non-smooth Lyapunov theory, in particular non-pathological Lyapunov functions and set-valued Lie derivatives (\cite{BacCer03}), for which we provide new results and novel proof techniques that are of independent interest. Due to the possible presence of discontinuities in the coupling function, the stability analysis is carried out by focusing on the regularization of the dynamics, as typically done in the hybrid formalism of \cite[Ch. 4]{TeelBook12}. Finally, simulations are provided to illustrate the theoretical guarantees and demonstrate the potential strength of our hybrid theoretical tools to address both first and second-order oscillators modeling generators in power grids considered in (\cite{dorfler2012synchronization}).

The recent submission (\cite{BossoLeaderKura}) (see also (\cite{BossoLeaderKuraCDC})) also uses hybrid tools to obtain uniform global synchronization guarantees in a Kuramoto setting but in a different context,  namely for second-order oscillators (where the $\omega_i$'s are states rather than external inputs) and, most importantly, for a network with a leader, which significantly changes the setting compared to the leaderless scenario investigated in this work, where no oscillator is insensible to the coupling actions from its neighbours. With respect to the preliminary version of this work in (\cite{IFACKura}), we include the next novel elements: relaxed requirements on the coupling function, time-varying, phase-dependent, (possibly) non-identical natural frequencies, generalizing the two-agents theorems of (\cite{IFACKura}) to the case of $n$ oscillators in addition to establishing a set of new stability results missing in (\cite{IFACKura}) (finite-time, practical properties and other ancillary results). 
 
The rest of the paper is organized as follows. Notation and background material are given in Section~\ref{sec:introduction}. The local hybrid coupling rules and oscillators network model are derived in Section~\ref{sec:model_disc}. In Section~\ref{sec:reg_mod}, we introduce the regularized version of the dynamics presented in Section~\ref{sec:model_disc}. In Section~\ref{sec:as_stability}, we present Lyapunov-based analysis tools establishing the asymptotic properties of our model, while prescribed finite-time results are given in Section~\ref{sec:fixed-time}. Numerical illustrations are provided in Section~\ref{sec:sims}, while most of the technical aspects of our proofs requiring non-smooth analysis concepts are gathered in Section~\ref{sec:proofs}. A few proofs of minor importance are relegated to the Appendix. 
\section{Preliminaries}
\label{sec:introduction}

\noindent
{\bf Notation}.
Let $\real:=(-\infty,\infty)$, $\real_{\geq0}:=[0,\infty)$, $\real_{>0}:=(0,\infty)$, $\integer_{\geq 0}:=\{0,1,\dots\}$, $\integer_{>0}:=\{1,2,\dots\}$ and $\integer_{>1}:=\{2,\dots\}$. The notation $\real^n$ denotes the $n$-dimensional Euclidean space with $n\in\integer_{>0}$ and $e_i$ is the $i$-th element of the natural base of $\real^n$, with $i \in \{1, \dots, n\}$. The notation $\mathbb{B}_n$ denotes the closed unit ball of $\real^n$ centered at the origin and we write $\mathbb{B}$ when its dimension is clear from the context. We denote with $\emptyset$ the \emph{empty set}. Given a vector $x\in\real^n$, we denote with $x_{\ell}$ its $\ell$-th element, with $\ell \in \{1,\dots, n\}$, $|x|$ is its Euclidean norm and $|x|_1$ is its 1-norm. The notation $\boldsymbol{0}_n$ denotes a vector  whose $n\in \integer_{>0}$ elements are all equal to $0$. The notation $\boldsymbol{1}_n$  denotes a vector whose $n\in \integer_{>0}$ elements are all equal to $1$. Given two vectors $x_1\in\real^n$ and $x_2\in\real^m$, we denote $(x_1,x_2):= [x_1^\top x_2^\top]^\top$. Given a matrix $A\in\real^{n \times m}$, $[A]_{\ell}$ stands for its $\ell$-th column and $(A)_{s}$ for its $s$-th row, where $\ell \in \{1,\dots, m\}$ and $s \in \{1,\dots, n\}$. Given a vector $x\in \real^n$ and a non-empty set $\A \subset \real^n$ with $n\in\integer_{>0}$, $|x|_\A :=\inf\{|x-y|:y\in \A \}$ is the distance of $x$ to $\A$. Given a set $\mathcal{S} \subset \real^n$,  cl($\mathcal{S}$) stands for its closure, $\partial\mathcal{S}$ is its boundary, int($\mathcal{S}$) is its interior and $\overline{\text{co}} \, \mathcal{S}$ is its closed convex hull. Given a finite set $\mathcal{S} \subset \real^n$,  $|\mathcal{S}|$ denotes its cardinal number. A function $f: \real^n \rightarrow \real_{\geq 0}$ is radially unbounded  if $f(x) \rightarrow \infty$  as $|x| \rightarrow \infty$. Let $f: \real^n \rightarrow  \real$ and $r \in \real$, we denote by $f(r)^{-1}$ the set $\{x \in \real^n : f(x) = r\}$. Let $X$ and $Y$ two non-empty sets,  $T:X \rightrightarrows Y$ denotes a \emph{set-valued} map from $X$ to  $Y$. We define the set-valued map $\sign: \real \rightrightarrows \{-1,1\}$ as $\sign(z)= -1$ when $z<0$, $\sign(z)=1$ when $z>0$ and $\sign(0)=\{-1,1\}$. We refer to class $\K$, $\K_\infty$ and $\KL$ functions as defined in \cite[Chap. 3]{TeelBook12}.
A function $f:\real \to \real$ is  \emph{piecewise continuous} if for any given interval $[a,b]$, with $a<b\in \real$, there exist a finite number of points $a\leq x_0<x_1<x_2<\dots<x_{k-1}<x_{k}\leq b$, with $k\in\integer_{\geq0}$  such that $f$ is continuous on $(x_{i-1},x_i)$ for any $i\in\{1,\dots,k\}$ and its one-sided limits exist as finite numbers. A function $f:\real \to \real$ is \emph{piecewise continuously differentiable} if  for any given interval $[a,b]$, with $a<b\in \real$, there exists a finite number of points $a\leq x_0<x_1<x_2<\dots<x_{k-1}<x_{k}\leq b$, with $k\in\integer_{\geq0}$  such that $f$ is continuous,  $f$ is continuously differentiable on $(x_{i-1},x_i)$ for any $i\in\{1,\dots,k\}$ and its one-sided limits of the difference quotient exist as finite numbers. We define with $\uni([a,b])$ the continuous uniform distribution over the compact interval $[a,b]$ with $a<b\in \real$.
\smallskip

\noindent
{\bf Background on graph theory}.
We denote an unweighted undirected graph as $\G_u = (\V, \E_u)$, where $\V$ is the set of vertices, or nodes, and $\E_u \subseteq \V \times \V$ is the set of edges, or arcs, composed by unordered pairs of nodes. If a pair $(i,j)$ of nodes belongs to $\E_u$, we say that those nodes are \emph{adjacent} and that $j$ is a \emph{neighbour} of $i$ and vice versa. Given two nodes $x$ and $y$ of an undirected graph $\G_u$, we define as \emph{path} from $x$ to $y$ a set of vertices starting with $x$ and ending with $y$, such that consecutive vertices are adjacent. If there is a path between any couple of nodes, the graph is called \emph{connected}, otherwise it is called \emph{disconnected}. We define as \emph{subgraph} of $\G_u$ a graph $\G_s = (\V_s, \E_s)$, where $\V_s \subset \V$ and $\E_s \subset \E_u$. An induced subgraph of $\G_u$ that is maximal, subject to be connected, is called a \emph{connected component} of $\G_u$. A \emph{cycle} is a connected graph where every vertex has exactly two neighbours.  An acyclic graph is a graph for which no subgraph is a cycle. A connected acyclic graph is called a \emph{tree}.

We denote an unweighted directed graph as $\G = (\V, \E)$, where $\E  \subseteq \V \times \V$ is composed of ordered pairs, therefore arcs have a specific direction. An arc going from node $i$ to node $j$ is denoted by $(i,j) \in \E$.
If a directed graph $\G$ is obtained choosing an arbitrary direction for the edges of an undirected graph $\G_u$, we call it an \emph{oriented} graph, and we say that $\G$ is obtained from an orientation of $\G_u$. If $(i,j) \in \E$, we say that $i$ belongs to the set of \emph{in-neighbors} $\mathcal{I}_j$ of $j$, while $j$ belongs to the set of \emph{out-neighbors} $\mathcal{O}_i$ of $i$. The union of $\mathcal{I}_i$ and $\mathcal{O}_i$ gives the more generic set of neighbors $\mathcal{V}_i := \mathcal{I}_i \cup \mathcal{O}_i$ of node $i$, containing all the nodes connected to it, in any direction. With $B \in \mathbb{R}^{n \times m}$ we denote the incidence matrix of graph $\G$ such that each column $[B]_\ell, \, \ell \in \{1,\ldots,m\}$, is associated to an edge $(i,j) \in \E$, and all entries of $[B]_\ell$ are zero except for $b_{i \ell} = -1$ (the tail of edge $\ell$) and $b_{j \ell} = 1$ (the head of edge $\ell$), namely $[B]_\ell = e_j - e_i$.

\section{Oscillators with hybrid coupling}
\label{sec:model_disc}
\subsection{Flow dynamics}
\label{subsec:flow_model_nagent}

Consider a networked system of $n$ heterogeneous oscillators. To achieve synchronization, the oscillators locally exchange coupling actions through the \emph{unweighted undirected tree}\footnote{As mentioned in the introduction, we can obtain a spanning tree using any of the distributed, finite-time algorithms described in (\cite{pandurangan2018distributed}).} $\G_{u}:= (\V, \E_u)$ made of $n$ nodes and thus $m=n-1$ edges,  $n\in \integer_{>1}$. We assign an arbitrary orientation to $\G_{u}$, which leads to the oriented tree $\G = (\V, \E)$. In this scenario, the oscillator phase corresponding to node $i$, with $i \in \mathcal{V}$, is denoted $\theta_i$ and has the next flow dynamics
\begin{align}
	\label{eq:flow_phase_nagent_unidrected}
	\dot{\theta}_i = \omega_i (\theta,\mathfrak{t}) + \kappa \sum_{j\in\mathcal{O}_{i}} \sigma(\theta_j - \theta_i + 2q_{ij}\pi)   - \kappa \sum_{j\in\mathcal{I}_{i}} \sigma(\theta_i - \theta_j + 2q_{ji}\pi), \quad (\theta,q) \in C
\end{align}
where $\omega_i (\theta,\mathfrak{t})$ is a possibly an unknown term modeling the dynamics of the  $i$-th oscillator, which can capture physical coupling actions, its time-varying natural frequency, and any other unknown bounded dynamics affecting the oscillator; see Section~\ref{sec:sims} for a numerical example. We assume that $\omega_i$  is locally bounded, measurable in $\mathfrak{t}$, piecewise continuous in $\theta$ and such that $\omega_i (\theta,\mathfrak{t}) \in \Omega:=[\omega_{\text{m}}, \omega_{\text{M}}]$ for any time $\mathfrak{t}\geq0$ and  $(\theta,q) \in C$, with $\omega_{\text{m}} \leq \omega_{\text{M}} \in  \real$, namely $\Omega$ is a compact interval of values\footnote{The assumption that $\omega_i$, for any $i\in\V$, takes values in the compact set $\Omega$ could be relaxed by only assuming boundedness of the mismatch $\sup\limits_{(\mathfrak{t}, \theta)  \in \real_{\geq 0} \times [-\pi - \delta, \pi + \delta] } |\omega_i (\mathfrak{t},\theta) - \omega_j (\mathfrak{t},\theta)|$ for any pair $(i,j) \in \E$, and adapting the proofs accordingly.}. 
Since \eqref{eq:flow_phase_nagent_unidrected} possibly has a discontinuous right-hand side, the notion of solution should be carefully defined, and we postpone this discussion to Section~\ref{sec:reg_mod}
(where we also prove the existence of solutions)
 to avoid overloading the exposition. 
For now it suffices to say that a function $\theta$ is a solution of \eqref{eq:flow_phase_nagent_unidrected} if it is absolutely continuous (i.e., it coincides with the integral of its derivative) and satisfies 
\eqref{eq:flow_phase_nagent_unidrected} almost everywhere. 

Phase $\theta_i$ in \eqref{eq:flow_phase_nagent_unidrected} evolves in the set $[-\pi-\delta,\pi+\delta]$, with $\delta \in (0,\pi)$, which thus covers the unit circle corresponding to phases taking values in $[-\pi,\pi]$. Parameter $\delta>0$ inflates the set of angles $[-\pi,\pi]$ to rule out Zeno solutions as explained in the following, see  Section~\ref{subsec:overall_model}. Thus, $\delta$ is a regularization parameter chosen to be  the same for each oscillator. Variable $q_{ij}$, with $(i,j) \in \E$, is a logic state taking values in $\{-1,0,1\}$, which is constant during flows. Its role is to unwind the difference between the two phases $ \theta_j$ and  $\theta_i$ through jumps. Indeed, since $\theta_j$ and $\theta_i$ are angles, to evaluate their mismatch, loosely speaking, we have to consider their minimum mismatch modulo $2 \pi$: $q_{ij}$ is introduced for this purpose as clarified in Section~\ref{subsec:jump_model_nagent}. The vectors $\theta$ and $q$ collect all the states $\theta_i$, $i \in \V$, and  $q_{ij}$, $(i,j) \in \E$, respectively,  as formalized in the following, together with the formal definition of the flow set $C$, namely a compact subset of the state-space where the solutions are allowed to evolve continuously. The gain $\kappa\in\real_{>0}$ is associated with the intensity of each coupling action and it is the same for each interconnection. Finally, the coupling action between each pair of nodes $(i,j) \in \E$ is defined as $\sigma(\theta_j - \theta_i + 2q_{ij}\pi)$, where $\sigma$ is the function used to penalize the phase mismatch $\theta_j - \theta_i + 2q_{ij}\pi$ between phases $\theta_j$ and $\theta_i$, and it satisfies the next property. 
\begin{property}
	\label{prop:sigma}
	Function $\sigma$ is piecewise continuous on $\domain \sigma
	:= [-\pi - \delta, \pi + \delta]$
	and satisfies
	\begin{enumerate}[label=\alph*), leftmargin=*, ref=\alph*]
		\item $\sigma(s) = -\sigma(-s)$ for any $ s \in \domain \sigma$,
		\label{prop:sigma_symm}
		\item there exists $\alpha \in \K$ such that $\sgn (s)\sigma(s) \geq \alpha(|s|)$ for any $ s \in \domain \sigma \setminus \{0\}$. \null \hfill $\square$
		\label{prop:sigma_sect}
	\end{enumerate}
\end{property}
Item~\ref{prop:sigma_symm}) of Property~\ref{prop:sigma} ensures that $\sigma$ is an  odd function and thus implies $\sigma(0)=0$, while item~\ref{prop:sigma_sect}) of Property~\ref{prop:sigma} guarantees that $\sigma(s)$ can only be zero at $s=0$. Notice that the sine function, customarily used in the classical Kuramoto model, satisfies item~\ref{prop:sigma_symm}) but not item~\ref{prop:sigma_sect}) of Property~\ref{prop:sigma}, which is fundamental to establish the global uniform stability result of this work. Examples of functions $\sigma$ satisfying  Property~\ref{prop:sigma} are depicted in Figure~\ref{fig:sigma}, together with the sine function for the sake of comparison. We emphasize that the mild assumptions of Property~\ref{prop:sigma} allow considering, among others, intuitive discontinuous selections such as the sign function of Figure~\ref{fig:sigma}, which leads to an interesting parallel between \eqref{eq:flow_phase_nagent_unidrected} and the ternary controllers considered in (\cite{DePersisFrasca}). Another possible example of $\sigma$ enjoying Property~\ref{prop:sigma} is $\sigma(s)=\sin(s)+u(s)$, where $u$ is such that items~\ref{prop:sigma_symm}) and \ref{prop:sigma_sect}) of Property~\ref{prop:sigma} hold. Note that when $u$ is negligible compared to $\sin$ in a neighborhood of the origin, the model behaves locally like the classical Kuramoto network. Also, Property 1 comes with no loss of generality as we consider the scenario where we have the freedom to design the coupling rules among the oscillators and thus $\sigma$.
\begin{figure}[ht!]
	\begin{center}
		\includegraphics[width=0.6 \columnwidth]{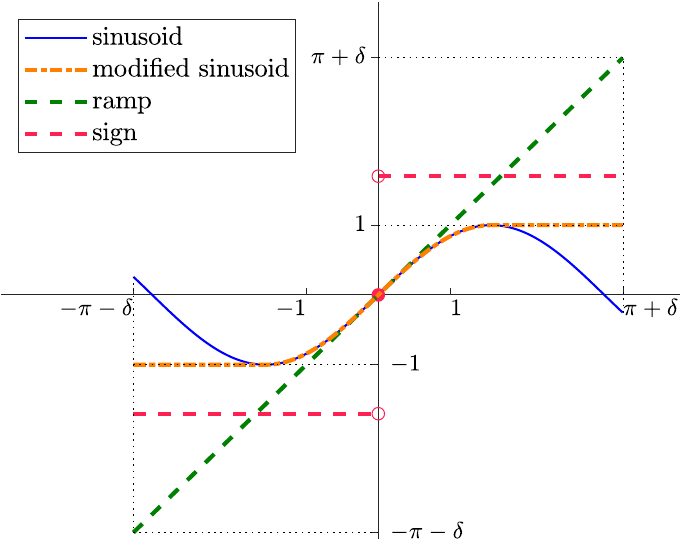}
	\end{center}
	\caption{Examples of functions $\sigma$ satisfying Property~\ref{prop:sigma}, together with the sine function (which does not satisfy Property~\ref{prop:sigma}).}
	\label{fig:sigma}
\end{figure}

Function $\sigma$ is only defined on $\domain \sigma = [- \pi - \delta, \pi+\delta]$ according to Property~\ref{prop:sigma}. 
We ensure in the sequel that the argument of $\sigma$ in \eqref{eq:flow_phase_nagent_unidrected}, namely $\theta_j - \theta_i +2q_{ij} \pi$, belongs to $\dom \sigma$ for all $(i,j)\in \E$, whenever $x \in C$, so that \eqref{eq:flow_phase_nagent_unidrected} is well-defined, see Section~\ref{subsec:jump_model_nagent}.

Collecting in the vector $\boldsymbol \sigma(x) \in \real^m$ all the coupling actions $\sigma(\theta_j - \theta_i + 2q_{ij}\pi)$, with $(i,j) \in \E$, using the same order as the columns of $B$, the flow dynamics in \eqref{eq:flow_phase_nagent_unidrected} is written as
\begin{align}
		\label{eq:flownD_compact}
		\dot{x} = 
		\begin{bmatrix}
			\dot \theta \\
			\dot q
		\end{bmatrix}
		= f(x,\boldsymbol{\omega}(\theta,\mathfrak{t}))&:=
		\begin{bmatrix}
			 \boldsymbol{\omega}(\theta,\mathfrak{t}) - B \kappa \boldsymbol{\sigma}(x) \\
			\boldsymbol{0}_m
		\end{bmatrix},
		& x \in C,
\end{align} 
with $\theta := (\theta_1, \ldots, \theta_n) \in [-\pi - \delta, \pi + \delta]^n$, $\boldsymbol{\omega}(\theta,\mathfrak{t}):= (\omega_1 (\theta,\mathfrak{t}), \ldots, \omega_n (\theta,\mathfrak{t}))\in\Omega^n$, and where $q \in \{-1,0,1\}^m$ is the vector stacking all the $q_{ij}$'s for $(i,j) \in \E$, ordered as in $\boldsymbol{\sigma}(x)$.
Thus, the overall state $x:=(\theta,q)$ evolves in the compact state space defined as
\begin{equation}
	X := [- \pi - \delta, \pi+\delta]^n \times \{-1, 0, 1\}^m.
	\label{eq:X}
\end{equation}
The flow set $C$ in \eqref{eq:flownD_compact} will be selected as the closed complement of the jump set $D$ introduced next.
\begin{figure*}[t]
	\includegraphics[width=0.32\textwidth]{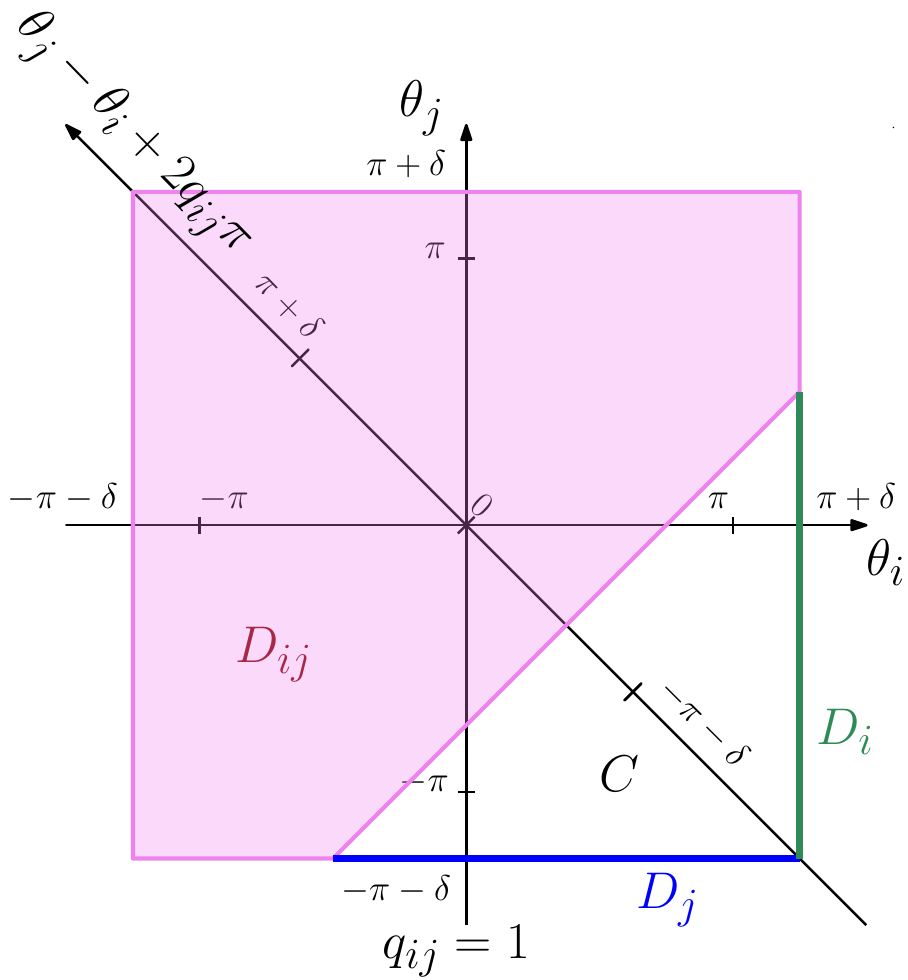} 
	\includegraphics[width=0.32\textwidth]{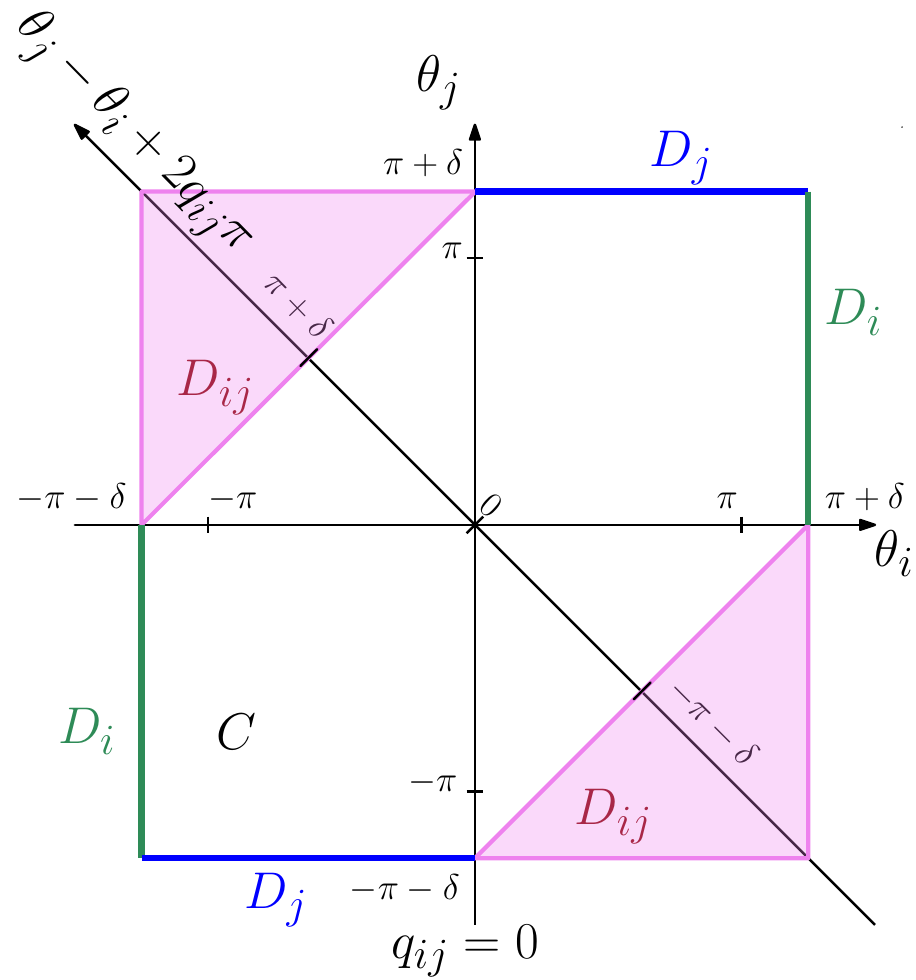} 
	\includegraphics[width=0.32\textwidth]{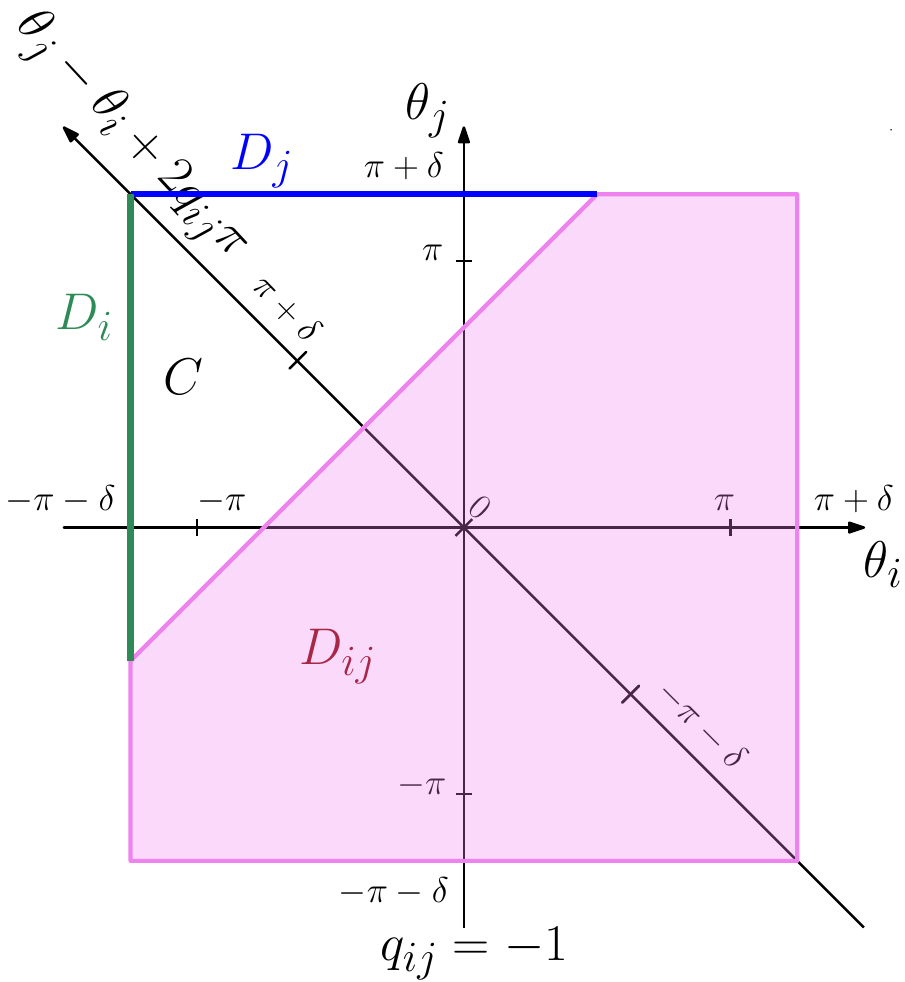}
	\caption{Projection of the flow and jump sets on $(\theta_i, \theta_j)$ for each value of $q_{ij}$.}
	\label{fig:statespace_ij}
\end{figure*}
\subsection{Jump dynamics}
\label{subsec:jump_model_nagent}
We introduce jump rules to constrain each phase $\theta_i$ to take values in $[- \pi - \delta, \pi+\delta]$ as well as to guarantee that the argument $\theta_j-\theta_i+2q_{ij}\pi$ of $\sigma$ in \eqref{eq:flow_phase_nagent_unidrected} belongs to  $ \dom \sigma = [- \pi - \delta, \pi+\delta]$ when flowing. To guarantee the latter property, define, for any $(i,j)\in \E$, the jump set
\begin{subequations}
	\label{eq:jump_rule_theta_ij}
\begin{align}
	\label{eq:Dij_multi}
	D_{ij} := \big\{ x \in X : |\theta_j - \theta_i + 2q_{ij}\pi| \geq \pi + \delta \big\},
\end{align}
and the associated difference inclusion
\begin{align}
	\label{eq:G_ij}
	x^+ = 	\begin{bmatrix}
		\theta^+ \\
		q^+
	\end{bmatrix}\in G_{ij}^\ext(x) &:=
	\begin{bmatrix}
		\theta \\
		G_{ij}(x)
	\end{bmatrix},
	&x \in D_{ij},
\end{align}
where the entries of $G_{ij}: X \rightrightarrows \{-1,0,1\}^m$ are given by
\begin{align}
	\label{eq:G_IJ}
	(G_{ij})_{(u,v)}\hspace{-1mm} :=\hspace{-1mm}
	\begin{cases}
		\hspace{-0.5mm}\argmin\nolimits\limits_{h \in \{-1,0,1\}} |\theta_j - \theta_i + 2h\pi|,\,\mbox{if}\, (u,v)=(i,j), \\
		\hspace{2mm}\{q_{uv}\},\hspace{27.5 mm}\mbox{otherwise},
	\end{cases}
\end{align}
\end{subequations}
with $(u,v),(i,j)\in \E$. Set $D_{ij}$ in \eqref{eq:Dij_multi} enforces a jump when $\theta_j - \theta_i + 2 q_{ij} \pi$ is not in $\domain \sigma$ for $(i,j)\in \E$. Across a jump, according to \eqref{eq:G_ij}, only $q_{ij}$ changes in such a way that $|\theta_j - \theta_i + 2 q_{ij} \pi| < \pi + \delta$ after a jump as formalized in the next lemma whose proof is given in Appendix~\ref{proof:jump_k_prop} to avoid breaking the flow of the exposition.
\begin{lemma}
	\label{lem:jump_k_prop}
	For any $(i,j)\in \E$ and $x \in D_{ij}$, any $x^+ \in G_{ij}^\ext(x)$ as per \eqref{eq:G_ij} satisfies $x^+ \in X$ and $|\theta_j^+ - \theta_i^+ + 2 q_{ij}^+ \pi| < \pi + \delta$. \null \hfill $\square$
\end{lemma}
 A second jump rule is introduced for when one of the oscillators  $i \in \V$ reaches $|\theta_i|=\pi+\delta$. In this case, a jump of $2\pi$ is enforced so that the phase then belongs to $(-\pi-\delta, \pi+\delta)$ while remaining the same modulo $2 \pi$. We define for this purpose
\begin{subequations}
	\label{eq:jump_rule_theta}
	\begin{align}
		\label{eq:g_Di}
			x^+ = 	\begin{bmatrix}
			\theta^+ \\
			q^+
		\end{bmatrix}= g_i(x) &:=
		\begin{bmatrix}
			g_{i,\theta}(x) \\
			g_{i,q}(x)
		\end{bmatrix},
		&x \in D_i, 
	\end{align}
	where the entries of $g_{i,\theta}: X \to [-\pi-\delta, \pi+\delta]^n$ and $g_{i,q}: X \to \{-1,0,1\}^m$ are defined as
	\begin{align}
		\label{eq:g_theta}
		(g_{i,\theta})_j &:=
		\begin{cases}
			\theta_i-\sgn(\theta_i)2\pi, &\mbox{if } j=i, \\
			\theta_j, & \mbox{otherwise},
		\end{cases} \\
		\label{eq:g_k}
		(g_{i,q})_{(u,v)} &:=
		\begin{cases}
			q_{uv} + \sgn(\theta_i), &\mbox{if } v=i, \\
			q_{uv} - \sgn(\theta_i), &\mbox{if } u=i, \\
			q_{uv}, &\mbox{otherwise},
		\end{cases}
	\end{align}
	with $j\in \V$ and $(u,v)\in \E$.
	The set $D_i,\ i \in \V$, is defined as
	\begin{align}
		\label{eq:Di_multi}
		D_i := \cl\big(\big\{ x \in X : x \notin D_{uv} &\mbox{ for any } (u,v)\in \E, \,\,\,\,\mbox{and } |\theta_i|=\pi+\delta\big\}\big).
	\end{align}
\end{subequations}
In view of \eqref{eq:Di_multi}, the jump rule \eqref{eq:g_Di}  is allowed when both $|\theta_i|=\pi+\delta$ and  $x$ is not in the interior of $D_{uv}$ for any $(u,v)\in \E$, where a jump may occur according to \eqref{eq:jump_rule_theta_ij}.

Note that each function $g_i$ is continuous on its (not connected) domain because $D_i$ does not contain points with $\theta_i = 0$ for any $i \in \V$.

Finally, switching/jumping ruled by \eqref{eq:jump_rule_theta} unwinds the phase $\theta_i$ without changing the phase mismatches between neighbours, defined as $(\theta_j - \theta_i + 2 q_{ij})$, as shown in the next lemma, whose proof is given in Appendix~\ref{proof:jump_theta_prop}.
\begin{lemma}
		\label{lem:jump_theta_prop}
		For each $i \in \V$ and $x \in D_i$,  $x^+ = g_i(x)$ implies $x^+ \in X$ and, for all $ (u,v) \in \E$,
		\begin{align}
			\label{eq:jump_theta_prop}
			\begin{cases}
				\theta_v^+ - \theta_u^+ + 2 q_{uv}^+ \pi =  \theta_v - \theta_u + 2 q_{uv} \pi, \\
				|\theta_i ^+|= \pi - \delta < \pi + \delta.
			\end{cases}
		\end{align}
  \null \hfill $\square$
\end{lemma}
\subsection{Overall model}
\label{subsec:overall_model}
In view of Sections~\ref{subsec:flow_model_nagent}-B, the overall hybrid model is given by
\begin{subequations}
	\label{eq:hybr_multi}
	\begin{equation}
	\left\{	
	\begin{aligned}	    
			&\dot x&&\mkern-22mu= f(x,\boldsymbol{\omega}(\mathfrak{t})), && \quad x \in C,\\
			&x^+&&\mkern-22mu\in G(x), && \quad x \in D,
	\end{aligned}
    \right.
	\end{equation}
	where $f$ is defined in \eqref{eq:flownD_compact}, and using \eqref{eq:Dij_multi} and \eqref{eq:Di_multi},
	\begin{align}
		\label{eq:D_multi}
		D &:= \bigg(\bigcup\limits_{i = 1}^{n}D_i\bigg) \cup \bigg(\bigcup\limits_{(i,j) \in \E}D_{ij}\bigg), \\
		\label{eq:C_multi}
		C &= \cl(X \setminus D),
	\end{align}
	with $X$ defined in \eqref{eq:X}.
	The set-valued jump map $G$ is defined in terms of its graph, which is given by
	\begin{align}
		\label{eq:G_multi}
		\gph{G} := \bigg(\bigcup\limits_{i = 1}^{n}\gph{g_i}\bigg) \cup \bigg(\bigcup\limits_{(i,j) \in \E}\gph{G_{ij}^\ext}\bigg),
	\end{align}
\end{subequations}
with $g_i$ and $G_{ij}^\ext$ as per \eqref{eq:G_ij}, \eqref{eq:g_Di}-\eqref{eq:g_k}. Figure~\ref{fig:statespace_ij} shows three projections of the state space $X$ on the plane $(\theta_i,\theta_j)$ for some $(i,j)\in \E$, which corresponds to a union of three squares, one for each value of $q_{ij}$. 
\begin{remark}
Since we envision engineering applications, each phase $\theta_i$ with $i \in \V$  may be reconstructed from the angular measurements provided by sensors. Due to the wide variety of outputs provided by commercial sensors, a relevant task is to extrapolate a continuous measurement from a sensor that may return values whose wrapping around $2\pi$ is unknown; see, for example, (\cite{reigosa2018permanent}) and (\cite{anandan2017wide}). In this scenario, we can implement an algorithm to extract a continuous measurement of the phase satisfying \eqref{eq:flownD_compact}. In particular, following a rationale similar to that proposed in \cite [Figure 1]{mayhew2012path} for a setting with sampled measurements, we may continuously update an estimate $\theta_{i,e}$ of $\theta_{i}$. Indeed, for each sensor output $\theta_{i,so}$, we may extract the lifted measurement as the closest one to $\theta_{i,e}$ when performing $2\pi$-wraps $\theta_{i,e}^+=\theta_{i,so} + 2 \pi \argmin\nolimits\limits_{h \in \{-1,0,1\}}|\theta_{i,so} - \theta_{i,e} + 2 h \pi|$. This rule parallels the selection of (15) and (27b) of (\cite{mayhew2012path}) for the simpler case of $\mathbb{S}^1$ and scalar angular measurements.
\end{remark}
\section{Regularized hybrid dynamics}
\label{sec:reg_mod}
\begin{subequations}
\label{eq:hybr_multi_reg_full}
Model \eqref{eq:hybr_multi} is a time-varying hybrid system with a possibly  discontinuous right-hand side, due to the mild properties of $\sigma$, see Property~\ref{prop:sigma}. Hence, solutions may be understood in the generalized sense of (\cite{TeelBook12}). We consider for this purpose the regularization of \eqref{eq:hybr_multi}, so that stability properties for the regularized system carry over to the nominal and generalized solutions of \eqref{eq:hybr_multi}. In particular, following \cite[Page 79]{TeelBook12} we consider 
\begin{equation}
\label{eq:hybr_multi_reg}	
\left\{	
\begin{aligned}	    
&\dot x&&\mkern-22mu\in F(x), && \quad x \in C,\\
&x^+&&\mkern-22mu\in G(x), && \quad x \in D,
\end{aligned}
\right.
\end{equation}
where $C$, $D$, and $G$ coincide with those in \eqref{eq:hybr_multi}, and the set-valued map $F$ regularizes $f$ in \eqref{eq:flownD_compact} as
\begin{align}
\label{eq:F_reg_full}
 F(x)&:=
\begin{bmatrix}
\, \widehat{\Omega} - B \kappa \widehat{\Sigma}(x) \\
	\boldsymbol{0}_m
\end{bmatrix}, \quad  \forall x\in X,
\end{align} 
\end{subequations}
with the sets $\widehat{\Omega}:=\Omega \times \dots \times \Omega=[\omega_{\text{m}},\omega_{\text{M}}]^n$ and $\widehat{\Sigma}$ being the Krasovskii regularization of the function $\boldsymbol{\sigma}$ in \eqref{eq:flownD_compact}, see for more details \cite[Page 4]{hajek1979discontinuous}. More specifically, following \cite[Def. 4.13]{TeelBook12}, $\widehat{\Sigma}(x):= \bigcap\limits_{s>0} \overline{\text{co}} \, \boldsymbol{\sigma}((x + s \mathbb{B})\cap C)$. It is readily verified that, denoting by $\widehat{\sigma}$ the Krasovskii regularization of the scalar function $\sigma$, namely
\begin{equation}
\label{eq:sigmahat}
\widehat{\sigma}(\tilde{\theta}):= \bigcap\limits_{s>0} \overline{\text{co}} \, \sigma([\tilde{\theta} - s, \tilde{\theta} + s]\cap [-\pi - \delta, \pi + \delta]),
\end{equation}
for any $\tilde \theta \in [-\pi-\delta, \pi+\delta]$, then the set-valued map $\widehat{\Sigma}(x)$ is the stacking (with the same ordering as in $\boldsymbol{\sigma}$) of the set-valued maps $\widehat{\sigma}_{ij}$ defined as
\begin{equation}
\label{thetatilda}
\widehat{\sigma}_{ij}:=\widehat{\sigma}(\tilde \theta_{ij}), \qquad \tilde \theta_{ij}:= \theta_j - \theta_i+2q_{ij}\pi,
\end{equation}
for all $(i,j)\in \E$. 

Since the jump set, flow set, and jump map of hybrid system \eqref{eq:hybr_multi_reg_full} coincide with those of \eqref{eq:hybr_multi}, and 
for any $x\in X$ and ${\boldsymbol{\omega}} \in \widehat \Omega$, $f(x, \boldsymbol{\omega}) \in F(x)$, we study the stability properties of solutions of \eqref{eq:hybr_multi} by concentrating on the regularized dynamics \eqref{eq:hybr_multi_reg_full}.
In addition to clarifying the nature of solutions of \eqref{eq:hybr_multi}, which may, among other things, present sliding behavior (see Section~\ref{sec:fixed-time}),
the advantage of using \eqref{eq:hybr_multi_reg_full} instead of \eqref{eq:hybr_multi} is that \eqref{eq:hybr_multi_reg_full} satisfies the so-called hybrid basic conditions \cite[As. 6.5]{TeelBook12}, which ensure its well-posedness \cite[Thm. 6.30]{TeelBook12}.
\begin{lemma}
	\label{lem:hbc}
	System \eqref{eq:hybr_multi_reg_full} satisfies the hybrid basic conditions (HBC) of \cite[As. 6.5]{TeelBook12}. \null \hfill $\square$
\end{lemma}
\begin{proof}
	Sets $C$ and $D$, as defined in \eqref{eq:Dij_multi}, \eqref{eq:Di_multi}, \eqref{eq:D_multi}, \eqref{eq:C_multi} are closed, as required by \cite[As. 6.5 (A1)]{TeelBook12}.
	On the other hand, $F$ is the Krasovskii regularization of a function, which satisfies the HBC in view of its locally boundedness on $C$ and \cite[Lemma 5.16]{TeelBook12} as shown in \cite[Ex.  6.6]{TeelBook12}, thus \cite[As. 6.5 (A2)]{TeelBook12} is satisfied.
	Lastly, each $g_i$ and $G_{ij}^\ext$ has a closed graph, and so due to \eqref{eq:G_multi} the graph of $G$ is closed as well.
	As consequence, according to \cite[Lemma 5.10]{TeelBook12}, $G$ is outer semicontinuous and it is also locally bounded relative to $D$, thereby satisfying \cite[As. 6.5 (A3)]{TeelBook12}.
\end{proof}

Among other useful properties, Lemma~\ref{lem:hbc} guarantees intrinsic robustness of the stability property established later in Sections~\ref{sec:as_stability} and \ref{sec:fixed-time}, see \cite[Ch. 7]{TeelBook12}.
To conclude this section, we note that all maximal solutions to \eqref{eq:hybr_multi_reg_full} are complete\footnote{A solution to a hybrid system $\mathcal{H}$ is \emph{maximal} if it cannot be extended and it is \emph{complete} if its domain is unbounded.} and exhibit a (uniform) average  dwell-time property, thereby excluding Zeno phenomena. 
We emphasize that, through Lemmas~\ref{lem:jump_k_prop} and \ref{lem:jump_theta_prop}, the parameter $\delta$ plays a key role in establishing that no complete discrete solution exists. In particular, the fact that $\delta\neq 0$ and $\delta\neq \pi$ is key for being able to exclude Zeno solutions.
\begin{proposition}
	\label{prop:t-comp}
	All solutions to \eqref{eq:hybr_multi_reg_full} enjoy a uniform average dwell-time property. Namely, there exist $\tau_D\in\real_{>0}$ and $J_0\in\integer_{\geq 0}$ such that, for any solution $x$ to \eqref{eq:hybr_multi_reg_full} and for any pair of hybrid times such that  $\mathfrak{t}+\mathfrak{j}\geq \mathfrak{s}+\mathfrak{r}$ with $(\mathfrak{s},\mathfrak{r}),(\mathfrak{t},\mathfrak{j})\in \dom x$, $\frac{1}{\tau_D}(\mathfrak{t}-\mathfrak{s}) + J_0 \geq (\mathfrak{j}-\mathfrak{r})$. Moreover, if $x$ is maximal, then it is $\mathfrak{t}$-complete, i.e., $\sup_\mathfrak{t} \dom x  = \sup\{\mathfrak{t} \in \real_{\geq 0}:  \exists \mathfrak{j} \in \integer_{\geq 0}, \, (\mathfrak{t},\mathfrak{j}) \in \dom x\}=+\infty$. \null \hfill $\square$
\end{proposition}
\begin{proof}
	We first recall that, in view of Lemma~\ref{lem:jump_theta_prop}, for any $i\in\V$, $x\in D_i$  and $x^+ = g_i (x)$ 
	\begin{equation}
		\label{eq:eps1}
		|\theta_i^+|= \pi - \delta < \pi + \delta,
	\end{equation}
	while the other $\theta_j$, $j \neq i\in\V$ remain unchanged across such a jump and so does $\theta_j-\theta_i + 2q_{ij}\pi$ for all  $(i,j)\in \E$. 
	We also recall that, from Lemma~\ref{lem:jump_k_prop}, for any $(i,j) \in \E$, $x\in D_{ij}$ and $x^+ \in G_{ij}^\ext(x)$
	\begin{equation}
	\label{eq:eps2}
	|\theta_j^+-\theta_i^+ + 2q_{ij}^+\pi| \leq \max(2\delta,\pi) < \pi + \delta,
	\end{equation}
	while the $\theta_u$ and the other $\theta_h-\theta_k + 2q_{hk}\pi$ remain unchanged for any $u\in\V$ and $(h,k)\neq(i,j) \in \E$.
  From uniform global boundedness of the right hand-side $F(x)$ of the flow dynamics (a consequence of the local boundedness of $F$ and of the boundedness of $X$), all solutions satisfy a global Lipschitz property with respect to the flowing time and \eqref{eq:eps1} and \eqref{eq:eps2} imply a uniform average dwell time on the jumps from $\bigcup\limits_{i = 1}^{n}D_i$ and $\bigcup\limits_{(i,j)\in \E}D_{ij}$, respectively.  Finally, the uniform average dwell time property of solutions jumping from $D$ derives directly from Lemma~\ref{lem:jump_theta_prop} and \eqref{eq:G_ij}-\eqref{eq:G_IJ}. Indeed, \eqref{eq:G_ij}-\eqref{eq:G_IJ} imply that jumping from $D_{ij}$ does not affect the triggering condition in $D_{u}$ or  $D_{hk}$, for any $u\in\V$ and $(i,j)\neq(h,k) \in \E$. In a similar way, Lemma~\ref{lem:jump_theta_prop} implies that jumping from $D_{i}$ does not affect the triggering condition in $D_{j}$ or  $D_{ij}$ or $D_{ji}$, for any $j\neq i \in\V$ and $(i,j)\in \E$ or $(j,i)\in \E$. Thus,  a uniform average dwell time on jumps from $D$ stems from the global Lipschitz property of the solutions  with respect to the flowing time, along with \eqref{eq:eps1} and \eqref{eq:eps2}. 
	
	We now check that maximal solutions of \eqref{eq:hybr_multi_reg_full} are complete by proving that the conditions of \cite[Prop. 6.10]{TeelBook12} hold.
	First, consider $\xi \in C \setminus D$. Because $\partial C \subset D$, then $\xi \in \mbox{int}(C)$. Therefore, there exists a neighbourhood $U$ of $\xi$ such that $U\subset C \setminus D$. Thus, for any $x \in U \subset C \setminus D$, the tangent cone\footnote{The \emph{tangent cone} to a set $S \subset \real^n$ at a point $x \in \real^n$, denoted $T_S(x)$, is the set of all vectors $w\in \real^n$ for which there exist $x_i \in S$, $\tau_i > 0$ with $x_i \rightarrow x$, and $\tau_i \searrow 0$, such that $w=  \displaystyle{\lim_{i \to \infty}\frac{x_i - x}{\tau_i}}$.} to $C$ at $x$ is $T_C (x)=\real^n \times \{ 0 \}^m$.
	Hence, any $\psi \in F(x)$ in \eqref{eq:F_reg_full} satisfies $\{\psi \} \cap  T_C(x) = \{\psi \}  \neq \emptyset$, so that
	\cite[Prop. 6.10 (VC)]{TeelBook12} holds for any  $\xi \in C \setminus D$.
	On the other hand, the state space $X$ in \eqref{eq:X} is bounded, thus item (b) in \cite[Prop. 6.10]{TeelBook12} is excluded.
	To rule out item (c) in \cite[Prop. 6.10]{TeelBook12}, from Lemmas~\ref{lem:jump_k_prop} and \ref{lem:jump_theta_prop}, we have $G(D) \subset X = C \cup D$.
	Hence, we can apply \cite[Prop. 6.10]{TeelBook12} to conclude that all maximal solutions are complete, thus obtaining $\mathfrak{t}$-completeness of solutions in view of their uniform average dwell time property established above.
\end{proof}
\section{Asymptotic stability properties}
\label{sec:as_stability}
\subsection{Synchronization set and its stability property}
\label{subsec:sync_set}
To analyze the synchronization properties of system \eqref{eq:hybr_multi_reg_full}, 
consider the set
\begin{equation}
	\A := \{x \in X : \theta_i = \theta_j + 2 q_{ij}\pi, \, \forall(i,j) \in \E\}.
	\label{eq:attractor_original}
\end{equation}
Because the network is a tree, for any $x\in\A$, the phases $\theta_i$  and  $\theta_j$ coincide modulo $2\pi$ not only  for any $(i,j) \in \E $ but also for any $i \in \V$ and $j \in \V \setminus \{i\}$. In other words, when $x \in \A$, all the oscillators are synchronized even if they do not share a direct link. We therefore call $\A$ the \emph{synchronization} set. Our main result below establishes a practical asymptotic stability result for $\A$, as a function of the coupling gain $\kappa$ appearing in the flow map \eqref{eq:flow_phase_nagent_unidrected}. The ``practical'' tuning of $\kappa$ depends on the following two parameters: 
\begin{equation}
\label{eq:pars}
\underline\lambda:= \lambda_{\min}(B^\top B), \, \overline{\omega}:=(n-1)|\omega_{\text{M}}-\omega_{\text{m}}|\geq \max\limits_{\substack{\widehat{\boldsymbol{\omega}} \in \widehat{\Omega}}}|B^\top \widehat{\boldsymbol{\omega}}|_1.  
\end{equation}
Parameter $\underline\lambda$ ensures a 
detectability property of the distance $|x|_{\A}$ 
in \eqref{eq:attractor_original} from the norm  $|\widehat{\boldsymbol{\sigma}}|$, for any $\widehat{\boldsymbol{\sigma}} \in \widehat{\Sigma}(x)$. In particular, we have from the results in (\cite{GodsilGraphTheory}). 
 \begin{lemma}
 	\label{lem:B_prop}
 	Since $\G$ is a tree, $\underline \lambda := \lambda_{\min}(B^\top B)>0$. \null \hfill $\square$
 \end{lemma}
\begin{proof}
Consider a tree graph $\G$ with incidence matrix $B \in \real ^{n \times n-1}$. By definition $\G$ has only one bipartite connected component. From \cite[Thm. 8.2.1]{GodsilGraphTheory}, $\rank(B)=n-1$ and thus $\dim(\nulls(B))=0$ by way of the fundamental theorem of linear algebra. Therefore, $By \neq 0$ for any $y\in\real^{n-1}\setminus\{\boldsymbol{0}_{n-1}\}$ which implies $y^\top B^\top B y = |B y|^2>0$ for any $y\in\real^{n-1}\setminus\{\boldsymbol{0}_{n-1}\}$. Consequently all the eigenvalues of $B^\top B$ are strictly positive, thus completing the proof.
\end{proof}
\begin{remark}
	\label{rem:lmbdamin}
	The smallest eigenvalue $\underline \lambda$ of $B^\top B$ and its positivity established in Lemma~\ref{lem:B_prop} play a fundamental role on the speed of convergence of the closed-loop solutions to the synchronization set. As $\G$ is a tree, positivity of $\underline \lambda$ is ensured by Lemma~\ref{lem:B_prop}. In more general cases with $\G$ not being a tree, the leaderless context considered in this paper, where the synchronized motion emerges from the network, poses significant obstructions to achieving global results. A simple insightful example of a cyclic graph is discussed in Section~\ref{subsec:counter}, which provides a clear illustration of the motivation behind requiring that $\G$ is a tree. We emphasize that a similar obstruction is experienced in prior work (\cite{HTsync}) where, in a different context, a similar assumption on the network is required.
\end{remark}
We are now ready to state the main result of this paper, corresponding to a practical $\KL$ bound on the distance of $x$ from $\A$ that is uniform in $\kappa$. We state the bound in our main theorem below, whose proof is given in Section~\ref{subsec:proofProp2}, and then illustrate its relevance on a number of corollaries given next.
\begin{theorem}
	\label{thm:practical_stability}
	Given set $\A$ in \eqref{eq:attractor_original}, there exists a class $\KL$ function $\beta_\circ$ and a class $\K$ gain $\gamma_\circ$, both of them independent of $\kappa$, such that, for any $\kappa>0$, all solutions $x$ of \eqref{eq:hybr_multi_reg_full} satisfy 
	\begin{equation}
	|x(\mathfrak{t},\mathfrak{j})|_{\A}\leq \beta_\circ(|x(0,0)|_{\A},\kappa\mathfrak{t}) + \gamma_\circ((\kappa\underline \lambda)^{-1}c\overline\omega),
	\label{eq:KLpract}
	\end{equation}
	 for all $(\mathfrak{t},\mathfrak{j})\in \dom x$ and with $c:=\max\limits_{\substack{s \in \dom \sigma}}\widehat \sigma(s)$. \null \hfill $\square$
\end{theorem}
The bound \eqref{eq:KLpract} in Theorem~\ref{thm:practical_stability} is the sum of two terms:  $\beta_\circ$ captures the phases tendency to synchronize, while function $\gamma_\circ$ depends on the mismatch among the (possibly) non-identical, time-varying natural frequencies of the oscillators, which hampers asymptotic phase synchronization in general. Therefore, Theorem~\ref{thm:practical_stability}  provides an insightful bound \eqref{eq:KLpract} illustrating the trend of the continuous-time evolution of the hybrid solutions to \eqref{eq:hybr_multi_reg_full}. Notice that $\beta_\circ$ and $\gamma_\circ$ can be constructed by following similar steps as the ones in \cite[Lemma 2.14]{SontagIOS}, noting that the resulting bound is often subject to some conservatism. On the other hand, because $\beta_\circ$ and $\gamma_\circ$ are independent of $\kappa$ and $\underline \lambda$, \eqref{eq:KLpract} still provides valuable quantitative information. Indeed, in view of \eqref{eq:KLpract},  increasing $\kappa$ speeds up the transient and reduces the asymptotic phase disagreement caused by the non-identical time-varying natural frequencies. Equation \eqref{eq:KLpract} also highlights the impact of the algebraic connectivity $\underline{\lambda}$ of  $\G$  \cite[pages 23-24]{mesbahi2010graph} on the phase synchronization, by giving information on the scalability of our algorithm. Recall that $\underline{\lambda}$ is influenced by several parameters of the undirected graph, such as the maximum degree and the number of nodes (\cite{rad2011lower}). This continuous-time focus in \eqref{eq:KLpract} in Theorem \ref{thm:practical_stability} is motivated by Proposition~\ref{prop:sigma}. It is also of interest to establish a bound similar to \eqref{eq:KLpract} while measuring the elapsed time in terms of $\mathfrak{t}+\mathfrak{j}$ and not only in terms of $\mathfrak{t}$, as usually done when defining bounds for solutions to hybrid systems, which allows us to ensure stronger stability properties, in particular, uniformity and robustness  \cite[Chp.7]{TeelBook12}. Hence, combining Theorem~\ref{thm:practical_stability} with Proposition~\ref{prop:t-comp}, we obtain the following second main result.
\begin{theorem}
	\label{thm:KLbound_t+j}
	For each value $\kappa>0$, there exists a class $\KL$ function $\beta$ such that all solutions to  \eqref{eq:hybr_multi_reg_full} satisfy 
	\begin{equation}
		|x(\mathfrak{t},\mathfrak{j})|_{\A}\leq \beta(|x(0,0)|_{\A},\mathfrak{t}+\mathfrak{j}) + \gamma_\circ((\kappa\underline \lambda)^{-1}c \overline\omega),
		\label{eq:KLpract_2}
	\end{equation}
for all $(\mathfrak{t},\mathfrak{j})\in \dom x$, with $\gamma_\circ$ as in Theorem~\ref{thm:practical_stability}. \null \hfill $\square$
\end{theorem}
\begin{proof}	
Let $\kappa>0$ and $x$ be a solution to \eqref{eq:hybr_multi_reg_full}.  In view of Proposition~\ref{prop:t-comp}, $\frac{1}{\tau_D}\mathfrak{t} + J_0 \geq \mathfrak{j}$ for any $(\mathfrak{t},\mathfrak{j})\in\dom x$, which is equivalent to $\frac{1}{2}\mathfrak{t} \geq \frac{\tau_D}{2}(\mathfrak{j}-J_0 )$. Hence, we derive from \eqref{eq:KLpract}, for any $(\mathfrak{t},\mathfrak{j})\in\dom x$, 
	\begin{align}
		    \nonumber
			&\beta_\circ(|x(0,0)|_{\A},\kappa\mathfrak{t})= \beta_\circ\left(|x(0,0)|_{\A},\kappa ( \frac{1}{2}\mathfrak{t} + \frac{1}{2}\mathfrak{t})\right)\\
			\nonumber
			&\leq\beta_\circ\left(|x(0,0)|_{\A},\kappa \max \left\{0,\frac{1}{2}\mathfrak{t} + \frac{\tau_D}{2}\mathfrak{j} - \frac{\tau_D}{2}J_0\right\}\right)\\
			\nonumber
			&\leq\beta_\circ \left(|x(0,0)|_{\A},\frac{\kappa}{2}\max\left\{0,\min(1,\tau_D)(\mathfrak{t} + \mathfrak{j}) - \tau_D J_0\right\}\right)\\
			&=:\beta(|x(0,0)|_{\A},\mathfrak{t} + \mathfrak{j}).
			\label{eq:bound_beta}
	\end{align}	
Function $\beta$ is of class $\KL$. Hence, \eqref{eq:KLpract} and \eqref{eq:bound_beta} yield \eqref{eq:KLpract_2}, thus completing the proof.
\end{proof}

Theorem~\ref{thm:KLbound_t+j} implies that the oscillator phases uniformly converge to any desired neighborhood of $\A$ by taking $\kappa$ sufficiently large, thus the practical nature of the result. We also immediately conclude from Theorem~\ref{thm:KLbound_t+j} and Lemma~\ref{lem:hbc} that the stability property in \eqref{eq:KLpract_2}  is robust in the sense of item (a) of \cite[Def. 7.18]{TeelBook12}, according to \cite[Thm. 7.21]{TeelBook12}.

We may draw an important additional conclusion from Theorem~\ref{thm:KLbound_t+j} 
corresponding to a global practical $\KL$ bound stemming from the fact that the 
function $\gamma_\circ$ in \eqref{eq:KLpract_2} is independent of $\kappa$.
\begin{corollary}
  \label{cor:practical_stability}  
Set $\A$ is uniformly globally practically $\KL$ asymptotically stable for system \eqref{eq:hybr_multi_reg_full}, i.e., for each $\varepsilon > 0$, there exists $\kappa^\star >0$ such that, for all $\kappa \geq \kappa^\star$, 
there exists $\beta\in\KL$ such that any solution $x$ verifies $|x(\mathfrak{t},\mathfrak{j})|_\A\leq \beta(|x(0,0)|_\A,\mathfrak{t}+\mathfrak{j}) + \varepsilon, \, \text{for all }(\mathfrak{t},\mathfrak{j})\in\dom x$. \null \hfill $\square$
\end{corollary}
Lastly, in the case of uniform frequencies $\boldsymbol{\omega}(\theta,\mathfrak{t})=\boldsymbol{1}_n \omega(\mathfrak{t})$, for all $\mathfrak{t} \geq 0$, with $\omega(\mathfrak{t})\in\Omega$, we have $B^\top \boldsymbol{\omega}(\mathfrak{t}) = 0$, for all $\mathfrak{t} \geq 0$.
Then, we can exploit the fact that the term $|\overline \omega|$ at the right-hand side of \eqref{eq:KLpract_2} 
stems from upper bounding $|B^\top \widehat{\boldsymbol{\omega}}|_1$ as in \eqref{eq:pars}, 
which allows obtaining the following asymptotic property of $\A$.
\begin{corollary}
	\label{cor:asympt_stability}
	If $\boldsymbol{\omega}(\theta,\mathfrak{t})=\boldsymbol{1}_n \omega(\mathfrak{t})$, for all $\mathfrak{t} \geq 0,$ with $\omega(\mathfrak{t}) \in \Omega$, then set $\A$ is uniformly globally $\KL$ asymptotically stable for system \eqref{eq:hybr_multi_reg_full}, i.e., for each $\kappa>0$, there exists $\beta\in\KL$ such that any solution $x$ verifies 
\begin{align}
      |x(\mathfrak{t},\mathfrak{j})|_\A\leq \beta(|x(0,0)|_\A,\mathfrak{t}+\mathfrak{j}), \quad \forall(\mathfrak{t},\mathfrak{j})\in\dom x. 
  \label{eq:KLbound}
  \end{align}  \null \hfill $\square$
\end{corollary}
\subsection{Cyclic graphs and their potential issues}
\label{subsec:counter}
Before proceeding with the technical derivations needed to prove Theorem~\ref{thm:practical_stability}, we devote some attention to the issues pointed out in Remark~\ref{rem:lmbdamin} about the need for the graph $\G$ to be a tree, similar to (\cite{HTsync}).
  
 Consider system \eqref{eq:hybr_multi} with $n=3$ and with $\G_u$ an all-to-all undirected  graph, thus not a tree. Let $\G$ be the orientation of $\G_u$ with the incidence matrix
  $B=
  \smallmat{
  -1 & 0 & 1\\1 & -1 & 0\\0 & 1 & -1
  }$.
We take $\delta=\frac{3}{4}\pi$, any $\kappa>0$, any $\sigma$ satisfying Property~\ref{prop:sigma}, and we select for convenience  $\boldsymbol{\omega}(\theta,\mathfrak{t})=\boldsymbol{0}_3$ for any time $\mathfrak{t}\geq 0$ and $x\in X$. Let $x=(\theta,q)$ be a solution to the corresponding system \eqref{eq:hybr_multi} initialized at $\left(-\frac{2}{3} \pi,0,\frac{2}{3} \pi,0,0,1\right)$. We have that $x(0,0)\in \text{int}(C) \setminus \A$ and $B^\top\theta(0,0)+2\pi q(0,0)=\frac{2}{3}\pi \boldsymbol 1_3$, thus implying  $\boldsymbol{\sigma}(x(0,0))=\sigma(\frac{2}{3}\pi) \boldsymbol 1_3$.  Hence, because  $B \boldsymbol 1_3 = 0$,  from \eqref{eq:flownD_compact} it holds that
$\dot{\theta}(0,0) 
=\boldsymbol{0}_3$, 
and consequently $\dom x\subset[0,\infty)\times\{0\}$, and  $x(\mathfrak{t},0)=x(0,0)$ and $x(\mathfrak{t},0)\in \text{int}(C)\backslash \A$ for all $(\mathfrak{t},0)\in\dom x$. As a result, solution $x$ does not converge to the synchronization set. 

More generally, when the graph is not a tree, the kernel of matrix $B$ contains additional elements besides the zero vector. Consequently, we can have $B\boldsymbol{\sigma}(x)=\boldsymbol{0}_n$ even when $x\notin \A$, and $\underline\lambda=0$ in \eqref{eq:pars}. As a result, $\A$ is not globally attractive. 
\subsection{A Lyapunov-like function and its properties}
\label{subsec:UGAS}
To prove Theorem~\ref{thm:practical_stability}, we rely on the Lyapunov function $V$, defined as
\begin{align}
	\label{eq:lyapFunc_practical_V}
	V(x) &:= \sum\nolimits\limits_{(i,j) \in \E} V_{ij}(x), \qquad \forall x\in X,\\
	\label{eq:lyapFunc_ij_practical}
	V_{ij}(x) &:= \int_{0}^{\theta_j - \theta_i + 2 q_{ij} \pi} \sigma(\sat_{\pi + \delta}(s)) ds,
\end{align}
with $\sat_{\pi + \delta}(s)$ given by
\begin{align*}
	\sat_{\pi + \delta}(s) := \max \big\{ \min \{ s, \pi + \delta \},-\pi - \delta  \big\}, \quad \forall s \in \real.
\end{align*}
Function $V$ enjoys useful relations with the distance of $x$ from the synchronization set $\A$, as formalized next. 

\begin{lemma} \label{lem:sandwich_etal}
Given function $V$ in \eqref{eq:lyapFunc_practical_V}-\eqref{eq:lyapFunc_ij_practical}, 
there exist $\alpha_1,\alpha_2  \in \K_\infty$ independent of $\overline \omega$ in \eqref{eq:pars} and of $\kappa$,  such that
\begin{align}
  &\alpha_1(|x|_{\A}) \leq V(x) \leq \alpha_2(|x|_{\A}), \quad\forall x\in X.
  \label{eq:sandwich_V}
\end{align}\null \hfill $\square$
\end{lemma}
\begin{proof}
	For each $x \in \A$, $V(x)=0$  in view of \eqref{eq:attractor_original}, \eqref{eq:lyapFunc_practical_V}, \eqref{eq:lyapFunc_ij_practical}, and for each $x\in (C\cup D) \setminus \A$, $V(x)>0$ in view of item~\ref{prop:sigma_sect}) of Property~\ref{prop:sigma}. In addition, $V$ is (vacuously) radially unbounded as $X$ is compact. Hence,  \eqref{eq:sandwich_V} holds from \cite[Page 54]{TeelBook12}.
\end{proof}

To prove Theorem~\ref{thm:practical_stability}, it is also fundamental to formalize the relation between the distance of $x$ from the set $\A$ and $\widehat \Sigma(x)$ in \eqref{eq:hybr_multi_reg_full}, as done in the next lemma.
\begin{lemma} \label{lem:lowerbound_sigma}
There exists a class $\K_\infty$ function $\eta$ such that, for each $x\in X$, $\eta(|x|_{\A}) \leq |\widehat {\boldsymbol{\sigma}}|^2, \quad \forall  \widehat {\boldsymbol{\sigma}} \in \widehat \Sigma(x)$. \null \hfill $\square$
\end{lemma}
\begin{proof}
From items~\ref{prop:sigma_symm}) and \ref{prop:sigma_sect}) of Property~\ref{prop:sigma}, $|\sigma (s)| \geq \alpha (|s|)$ for any $s \in  \dom \sigma $. Thus, in view of \eqref{eq:sigmahat}, $|\varsigma|\geq \alpha(|s|)$, for any $s \in  \dom \sigma$ and $\varsigma \in \widehat{\sigma}(s)$. We recall that, for any $x \in X$, $\widehat{\Sigma}(x)$ is the stacking of all the set-valued maps $\widehat{\sigma}_{ij}$  defined in \eqref{thetatilda}, $(i,j) \in \E$. Hence, by definition of $\A$, for any $x\in X \setminus \A$, there exists at least one element $\tilde{\theta}_{ij} \neq 0$, with $(i,j) \in \E$, thus  $\widehat{\boldsymbol{\sigma}} \in \widehat{\Sigma}(x)$ implies that  $|\widehat{\boldsymbol{\sigma}}| \geq |\widehat{\sigma}_{ij}(x)| \geq \alpha(|\tilde{\theta}_{ij}|)$. Similarly, for any $x\in\A$, $|\widehat{\boldsymbol{\sigma}}| \geq |\widehat{\sigma}_{ij}(x)| \geq \alpha(|\tilde{\theta}_{ij}|)= 0$. Therefore, $\max_{(i,j)\in \E}\alpha(|\tilde{\theta}_{ij}|)$ is a suitable lower bound for $|\widehat{\boldsymbol{\sigma}}|$, for any $x\in X$. Since $\tilde{\theta}_{ij}$ is a function of the states and $\max_{(i,j)\in \E} \alpha(\cdot)$ is positive definite and radially unbounded, as $X$ is compact, then \cite[Page 54]{TeelBook12} implies that there exists $\eta \in \K_\infty$ such that $\eta(|x|_{\A}) \leq |\widehat {\boldsymbol{\sigma}}|^2$ holds for each $x\in X$ and for all $\widehat {\boldsymbol{\sigma}} \in \widehat \Sigma(x)$, thus concluding the proof.
\end{proof}

Function $V$ is locally Lipschitz due to the properties of $\sigma$ and characterizing its variation when evaluated along the solutions of the hybrid inclusion \eqref{eq:hybr_multi_reg_full} requires using tools from non-smooth analysis. To avoid breaking the flow of the exposition, we postpone to Section~\ref{subsec:proofProp2} those technical derivations and summarize the corresponding conclusions in the next proposition, a key result for proving Theorem~\ref{thm:practical_stability}.
\begin{proposition}
	\label{prop:PropV_practical}
Consider system \eqref{eq:hybr_multi_reg_full} and function $V$ in \eqref{eq:lyapFunc_practical_V}-\eqref{eq:lyapFunc_ij_practical}. There exist $\alpha_3 \in \K_\infty$ independent of $\overline \omega$ in \eqref{eq:pars} and of $\kappa$, such that for any $\kappa>0$, any solution $x$ of \eqref{eq:hybr_multi_reg_full} satisfies (denoting $\dom x=\bigcup\nolimits\limits_{\mathfrak{j}=0}^{J} [\mathfrak{t}_\mathfrak{j}, \mathfrak{t}_{\mathfrak{j}+1}] \times \{\mathfrak{j}\}$, possibly with $J=+\infty$)\\
(i) for all $\mathfrak{j} \in \{0, \ldots, J\}$ and almost all $\mathfrak{t} \in [\mathfrak{t}_\mathfrak{j}, \mathfrak{t}_{\mathfrak{j}+1}]$,
  \begin{subequations} 
\begin{align}
 &\frac{d}{d\mathfrak{t}} V(x(\mathfrak{t},\mathfrak{j})) \leq - \kappa\underline \lambda\alpha_3(V(x(\mathfrak{t},\mathfrak{j}))) + c\overline\omega,
 \label{eq:flowCond_traj}
\end{align}
with $c$ defined in Theorem~\ref{thm:practical_stability};\\
(ii) for all $\mathfrak{j} \in \{0, \ldots, J-1\}$,
\begin{align}
&V(x(\mathfrak{t},\mathfrak{j}+1)) \leq V(x(\mathfrak{t},\mathfrak{j})). 
\label{eq:jumpCond_traj}
\end{align}
\end{subequations}\null \hfill $\square$
\end{proposition}
We are now ready to present the proof of Theorem~\ref{thm:practical_stability}.

\begin{proofname}{Proof of Theorem~\ref{thm:practical_stability}:}
\label{subsec:proof_UGPAS}
Let $\kappa>0$ and $x$ be  a solution \eqref{eq:hybr_multi_reg_full} and denote, with a slight abuse of notation, $\dom x=\bigcup\nolimits\limits_{\mathfrak{j}=0}^{J} [\mathfrak{t}_\mathfrak{j}, \mathfrak{t}_{\mathfrak{j}+1}] \times \{\mathfrak{j}\}$ with $J\in \integer_{\geq 0} \cup \{+\infty\}$. We scale the continuous time as $\tau:= (\kappa\underline \lambda)\mathfrak{t} \in \real_{\geq 0}$ and we denote $\tau_\mathfrak{j}:= (\kappa\underline \lambda)\mathfrak{t}_\mathfrak{j}$ and  $V'(x(\cdot,\cdot))$ the time-derivative of $V$ with respect to $\tau$. From item~(i) of Proposition~\ref{prop:PropV_practical}, for all  $\mathfrak{j} \in \{0, \ldots, J\}$ and almost all $\tau \in [\tau _\mathfrak{j}, \tau _{\mathfrak{j}+1}]$,  
	\begin{align} 
		V'(x(\mathfrak{\tau},\mathfrak{j}))&=(\kappa\underline \lambda)^{-1} \dot V(x(\mathfrak{t},\mathfrak{j}))  \leq - \alpha_3(V(x(\mathfrak{t},\mathfrak{j}))) +(\kappa\underline \lambda)^{-1}c\overline\omega.
		\label{eq:flowCond_practical_tau}
	\end{align}
Combining \eqref{eq:flowCond_practical_tau} with the non-increase condition in \eqref{eq:jumpCond_traj}, we follow the steps of the proof of \cite[Lemma 2.14]{SontagIOS} to obtain an input-to-state stability bound on $V(\tilde x(\cdot,\cdot))$ where $\tilde x(\cdot,\cdot):=x((\kappa\underline \lambda)^{-1}(\cdot),(\cdot))=x(\cdot,\cdot)$, which can then be converted to a bound on $|x(\cdot,\cdot)|_\A$ using \eqref{eq:sandwich_V}, thus leading to \eqref{eq:KLpract}, where $\beta_\circ$ and $\gamma_\circ$ only depend on $\alpha_1$, $\alpha_2$ and $\alpha_3$ and are therefore independent of $\overline \omega$ and $\kappa$. Note that the dependence on $\underline \lambda$ is left implicit in $\beta_\circ$ and $\gamma_\circ$. 
\end{proofname}

\section{Prescribed finite-time stability properties}
\label{sec:fixed-time}
A useful outcome of the mild regularity conditions that we require from $\sigma$ in Property~\ref{prop:sigma} is that defining $\sigma$ to be discontinuous at the origin,
as in the sign function represented in Figure~\ref{fig:sigma}, leads to desirable sliding-like behavior of the solutions in the attractor $\A$. 
This sliding property induces interesting advantages of the behavior of solutions, as compared to the general asymptotic and practical properties characterized in Section~\ref{sec:as_stability}.

A first advantage is that, even with non-uniform natural frequencies, we prove uniform global $\KL$ asymptotic stability of $\A$ for a large enough coupling gain $\kappa$, due to the well-known ability of sliding-mode mechanisms to dominate unknown additive bounded disturbances acting on the dynamics. A second advantage is that the Lyapunov decrease characterized in Proposition~\ref{prop:PropV_practical} can be associated with a guaranteed constant negative upper bound outside $\A$, which implies finite-time convergence. Finally, since this constant upper bound can be made arbitrarily negative by taking $\kappa$ sufficiently large, we actually prove prescribed finite-time convergence 
(see (\cite{Krstic17})) when using these special discontinuous functions $\sigma$, whose peculiar features are characterized in the next lemma.

\begin{lemma} \label{lem:disc_sigma_lower_b}
Given a function $\sigma$ satisfying Property~\ref{prop:sigma}, if $\sigma$ is discontinuous at the origin, then there exists $\mu >0$ such that, for any $x \in X \setminus \A$, $|\widehat {\boldsymbol{\sigma}}| \geq \mu$, for all $\widehat {\boldsymbol{\sigma}} \in \widehat \Sigma(x)$. \null \hfill $\square$
\end{lemma}
 
\begin{proof}
Since $\sigma$ is discontinuous at 0 and it is piecewise continuous, there exists $\varepsilon>0$ such that $\sigma$ is continuous in $[-\varepsilon, 0)$ and $(0, \varepsilon]$. By item~\ref{prop:sigma_symm}) of Property~\ref{prop:sigma}, $\lim\limits_{s \to 0^+} \sigma(s)=-\lim\limits_{s \to 0^-} \sigma(s)=:\sigma_{\circ}\neq 0$ as $\sigma$ is discontinuous at 0 and $\sigma(0)=0$. Then there exists $\varepsilon_\circ \in (0, \varepsilon]$ such that $\sigma (s) \geq \frac{\sigma_{\circ}}{2}$ for all $s\in (0, \varepsilon_\circ]$. From item~\ref{prop:sigma_sect}) of Property~\ref{prop:sigma}, for any $s \in [\varepsilon_\circ, \pi + \delta]$, $\sigma (s) \geq \alpha (\varepsilon_\circ)>0$. Hence, due to item~\ref{prop:sigma_symm}) of Property~\ref{prop:sigma}, $|\sigma (s)| \geq \mu:=\min(\frac{\sigma_\circ}{2},\alpha (\varepsilon_\circ))$ for all $s \in  \dom \sigma \setminus \{0\}$. Moreover, in view of \eqref{eq:sigmahat}, for any $s \in  \dom \sigma \setminus \{0\}$ and any $\varsigma \in \widehat{\sigma}(s)$, $|\varsigma|\geq \mu$. Since, for any $x \in X$, $\widehat{\Sigma}(x)$ is the stacking of all the set-valued maps $\widehat{\sigma}_{ij}=\widehat{\sigma}$, $(i,j) \in \E$, and by definition of $\A$, for any $x\in X \setminus \A$, there exists at least one nonzero element $\tilde{\theta}_{ij} \neq 0$ for some $(i,j) \in \E$. Then  $\widehat{\boldsymbol{\sigma}} \in \widehat{\Sigma}(x)$ implies $|\widehat{\boldsymbol{\sigma}}| \geq |\widehat{\sigma}_{ij}(x)| \geq \mu$, thus concluding the proof. 
\end{proof}

Paralleling the structure of Proposition~\ref{prop:PropV_practical}, the next proposition, whose proof is postponed to Section~\ref{subsec:proofProp3}, is a key result for proving Theorem~\ref{thm:fixed-time}.

\begin{proposition}
\label{prop:PropV_finite_time}
	Consider system \eqref{eq:hybr_multi_reg_full} and function $V$ in \eqref{eq:lyapFunc_practical_V}-\eqref{eq:lyapFunc_ij_practical}. If $\sigma$ is discontinuous at the origin, then there exist $\mu \in \real_{>0}$ independent of $\overline \omega$ in \eqref{eq:pars} and $\kappa^\star>0 $ such that for each $\kappa \geq \kappa^\star$ any solution $x$ of \eqref{eq:hybr_multi_reg_full} satisfies (denoting $\dom x=\bigcup\nolimits\limits_{\mathfrak{j}=0}^{J} [\mathfrak{t}_\mathfrak{j}, \mathfrak{t}_{\mathfrak{j}+1}] \times \{\mathfrak{j}\}$, possibly with 
	$J = +\infty$)\\
	(i) for all $\mathfrak{j} \in \{0, \ldots, J\}$ and almost all $\mathfrak{t} \in [\mathfrak{t}_\mathfrak{j}, \mathfrak{t}_{\mathfrak{j}+1}]$ such that $x(\mathfrak{t},\mathfrak{j}) \notin \A$,
	\begin{subequations} 
		\begin{align}
			&\frac{d}{d\mathfrak{t}} V(x(\mathfrak{t},\mathfrak{j})) \leq -\frac{1}{2}\kappa \underline\lambda \mu^2;
			\label{eq:flowCond_traj_finite_time}
		\end{align}
		(ii) for all $\mathfrak{j} \in \{0, \ldots, J-1\}$,
		\begin{align}
			&V(x(\mathfrak{t},\mathfrak{j}+1)) \leq V(x(\mathfrak{t},\mathfrak{j})). 
			\label{eq:jumpCond_traj_finite_time}
		\end{align}
	\end{subequations} \null \hfill $\square$
\end{proposition}
Exploiting Lemma~\ref{lem:disc_sigma_lower_b} and Proposition~\ref{prop:PropV_finite_time}, we can follow similar steps to those in the proof of Theorem~\ref{thm:practical_stability} to show the following main result on uniform global $\KL$ asymptotic stability and
prescribed finite-time stability of $\A$ for \eqref{eq:hybr_multi_reg_full}.
\begin{theorem}
\label{thm:fixed-time}
If $\sigma$ is discontinuous at the origin, then set $\A$ in \eqref{eq:attractor_original} is prescribed finite-flowing-time stable for \eqref{eq:hybr_multi_reg_full}, i.e., for each $T>0$ there exists $\kappa^\star>0$ such that for each $\kappa \geq \kappa^\star$:\\
(i) there exists $\beta \in \KL$ such that all solutions $x$ satisfy \eqref{eq:KLbound};\\
(ii) all solutions $x$ satisfy, $x(\mathfrak{t},\mathfrak{j}) \in \A$ for all $(\mathfrak{t},\mathfrak{j}) \in \dom x$ with $\mathfrak{t}\geq T$. \null \hfill $\square$
\end{theorem}
\begin{proof}
We start showing that $G(D\cap \A)\subset\A$. Indeed, we notice that $D_{ij}\cap \A =\emptyset$ for any $(i,j)\in\E$, and thus $G(D_{ij}\cap \A)\subset\A$ trivially holds. Moreover, from Lemma~\ref{lem:jump_theta_prop},  it holds that $G(D_{i}\cap \A)\subset\A$ for any $i\in\V$. Hence, from \eqref{eq:D_multi}, we conclude that $G(D\cap\A)\subset\A$. To establish that $\A$ is (strongly) forward invariant for \eqref{eq:hybr_multi_reg_full}, it is left to prove that solutions cannot leave $\A$ while flowing. We proceed by contradiction and for this purpose suppose there exists a solution $x_{\text{bad}}$ to \eqref{eq:hybr_multi_reg_full} such that $x_{\text{bad}}(0,0)\in\A$ and  $x_{\text{bad}}(\mathfrak{t}^*,0)\notin\A$ for some $\mathfrak{t}^*>0$ with $(\mathfrak{t}^*,0) \in \dom x_{\text{bad}}$. From continuity of flowing solutions between any two successive jumps  and closedness of $\A$, there exists $x_{\text{bad}} (\mathfrak{t},0) \in \A$ for all $\mathfrak{t} \in [0,{\mathfrak{t}}^*)$ and $x_{\text{bad}} ({\mathfrak{t}}^*,0) \notin \A$. Hence, from \eqref{eq:lyapFunc_practical_V} and \eqref{eq:lyapFunc_ij_practical} and positive definiteness of $V$,  we have  $0=V(x_{\text{bad}} (0,0))<V(x_{\text{bad}} ({\mathfrak{t}}^*,0))$. However, since the solution is flowing, integrating  \eqref{eq:flowCond_traj_finite_time} over the continuous time interval $[0, {\mathfrak{t}}^*]$ we obtain $V(x_{\text{bad}} ({\mathfrak{t}}^*,0))<V(x_{\text{bad}} (0,0))$, which establishes a contradiction. Consequently, a solution cannot leave $\A$ while flowing. We have proven that the set $\A$ is (strongly) forward invariant, implying that if $x(\mathfrak{t},\mathfrak{j}) \in X \setminus \A$ then $x(\mathfrak{t'},\mathfrak{j'}) \in X \setminus \A$, for any $\mathfrak{t'}+\mathfrak{j'} \leq \mathfrak{t}+\mathfrak{j}$, with $(\mathfrak{t'},\mathfrak{j'}),(\mathfrak{t},\mathfrak{j})\in \dom x$. Let $\kappa \geq \kappa ^\star$ with $\kappa ^\star$ defined in  Proposition~\ref{prop:PropV_finite_time} and $x$ be a solution to \eqref{eq:hybr_multi_reg_full}. Combining \eqref{eq:flowCond_traj_finite_time} with the non-increase condition \eqref{eq:jumpCond_traj_finite_time} and the forward invariance of $\A$, we obtain  by integration for any $(\mathfrak{t},\mathfrak{j})\in \dom x$
 \begin{align}
	\label{eq:V_fin_time}
	\begin{split}
		V(x(\mathfrak{t},\mathfrak{j}))\leq -\frac{1}{2}\kappa \underline\lambda \mu^2 \mathfrak{t} + V(x(0,0)),
	\end{split}
\end{align}
whenever $x(\mathfrak{t},\mathfrak{j}) \in X \setminus \A$, and thus 
 \begin{align}
	\label{eq:V_fin_time_bound}
	\begin{split}
		V(x(\mathfrak{t},\mathfrak{j}))\leq \max(-\frac{1}{2}\kappa \underline\lambda \mu^2 \mathfrak{t} + V(x(0,0)),0),
	\end{split}
\end{align}
for any $x(\mathfrak{t},\mathfrak{j}) \in X$.
Equation \eqref{eq:V_fin_time_bound} can then be converted to a bound on $|x(\mathfrak{t},\mathfrak{j})|_\A$ using \eqref{eq:sandwich_V}. Hence, we follow the same steps used in the proofs of Theorem~\ref{thm:practical_stability} and \ref{thm:KLbound_t+j} to obtain \eqref{eq:KLbound}, thus concluding the proof of item~(i) in Theorem~\ref{thm:fixed-time}. 
In view of \eqref{eq:V_fin_time} and from the positive definiteness of $V$ with respect to $\A$, we conclude that  solutions to \eqref{eq:hybr_multi_reg_full} reach the synchronization set $\A$ flowing at most for $T:=\frac{ 1}{\kappa} \frac{2 \overline{v}}{\underline \lambda \mu^2}$, where $\overline{v}:= \max\limits_{x \in X} V(x)$. Hence, in view of the forward invariance of $\A$,  item~(ii) in Theorem~\ref{thm:fixed-time} holds thus completing the proof.  
\end{proof}

Notice that phases synchronize at most at continuous time  $T=\frac{ 1}{\kappa} \frac{2 \overline{v}}{\underline \lambda \mu^2}$, in view of \eqref{eq:V_fin_time_bound}. Therefore, we may decrease $T$ at will by selecting a larger $\mu$ as in Lemma~\ref{lem:disc_sigma_lower_b} and/or by increasing the coupling gain $\kappa$.

\begin{figure}[ht]
	\begin{center}
		\includegraphics[trim={2cm 1.5cm 1.5cm 1.5cm}, clip,width=0.8\columnwidth]{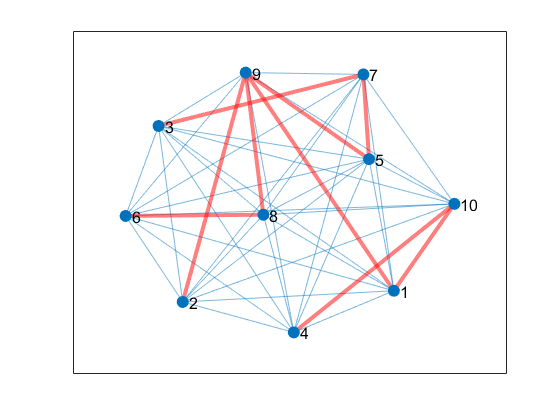}
	\end{center}
	\vspace*{-0.5cm}
	\caption{The networks of oscillators described in Section~\ref{sec:sims}: the blue graph depicts physical couplings captured by $\omega_i$ in \eqref{eq:wsims}, and the red graph is the communication tree graph that we design for the hybrid coupling rules presented in Section~\ref{sec:model_disc}.}
	\label{fig:network}
\end{figure}

\begin{figure*}[t]
\begin{minipage}[b]{0.32\textwidth}
\centering
        \includegraphics[trim={0.1cm 2.0cm 1.1cm 1.9cm},width=1\textwidth,clip]{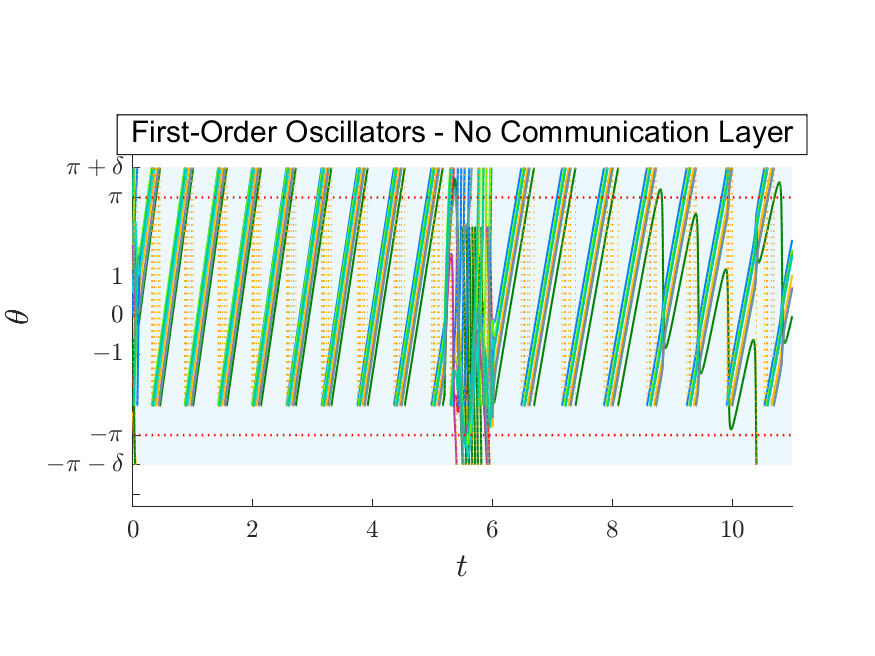} \\  
         \includegraphics[trim={2.5cm 0cm 5.5cm 0cm},width=0.363\textwidth,clip]{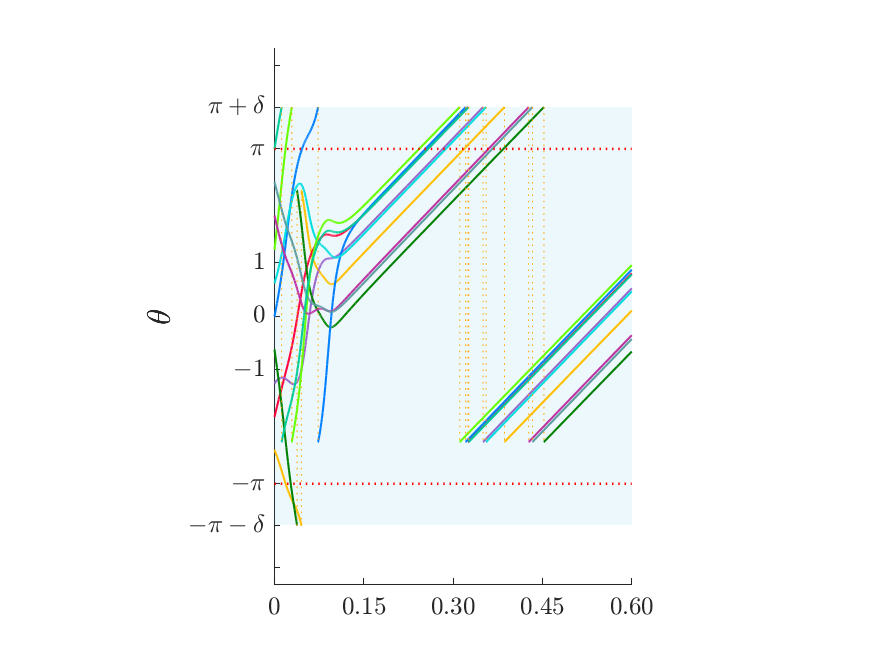}   
        \includegraphics[trim={4.5cm 0.167cm 4.8cm 0.1cm},width=0.301365\textwidth,clip]{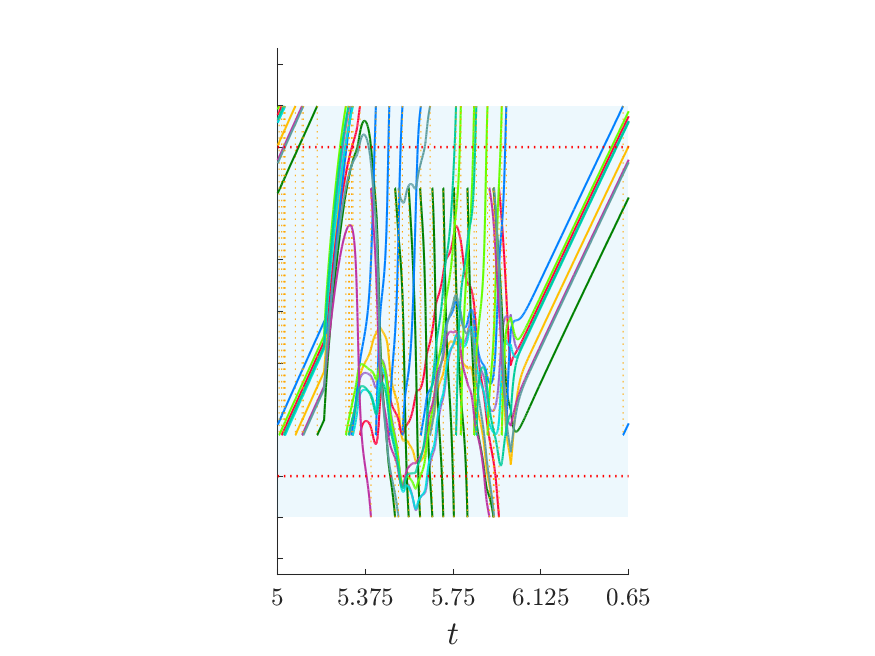} 
        \includegraphics[trim={4.5cm 0cm 5cm 0cm},width=0.28638\textwidth,clip]{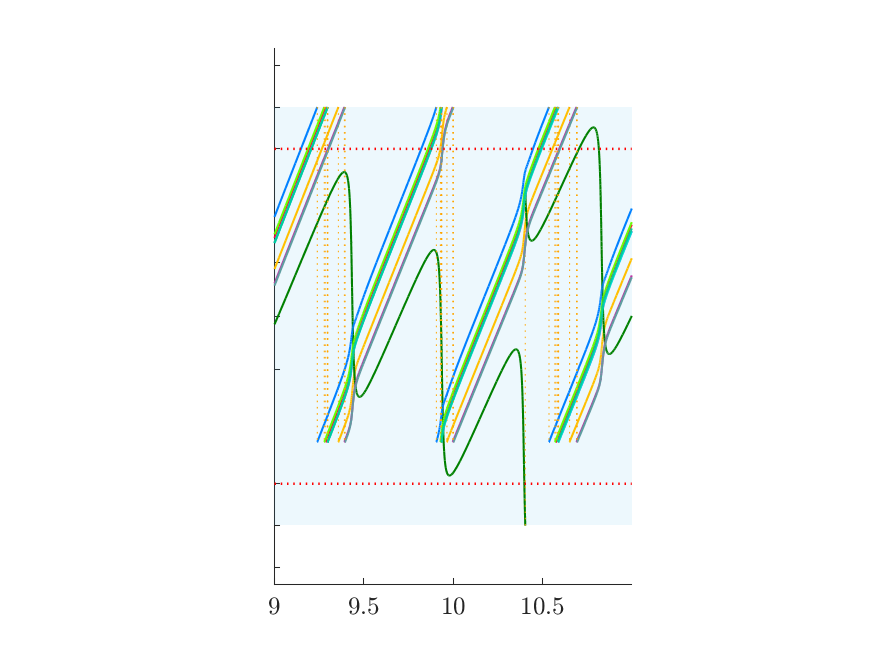}  \\  
        \includegraphics[trim={1.6cm 0.1cm 2.6cm 0.1cm},width=0.85\textwidth,clip]{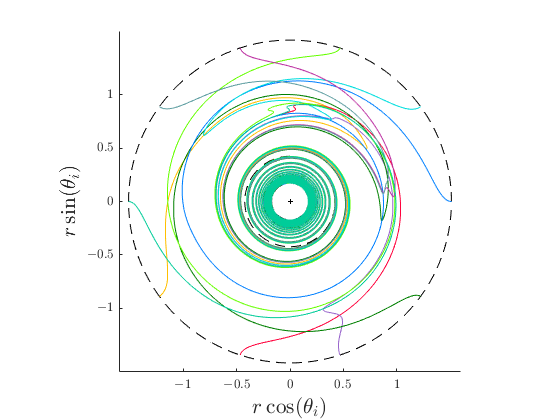} 
\end{minipage}
\begin{minipage}[b]{0.32\textwidth}
\centering
\includegraphics[trim={0.1cm  2.0cm 1.1cm 1.9cm},width=1\textwidth,clip]{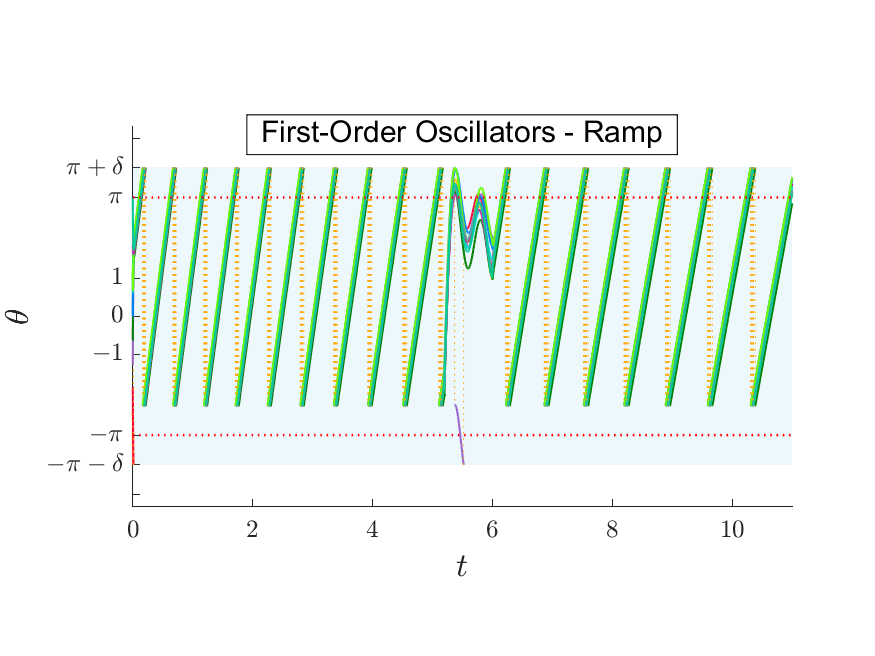} \\
                 \includegraphics[trim={2.5cm 0cm 5.5cm 0cm},width=0.363\textwidth,clip]{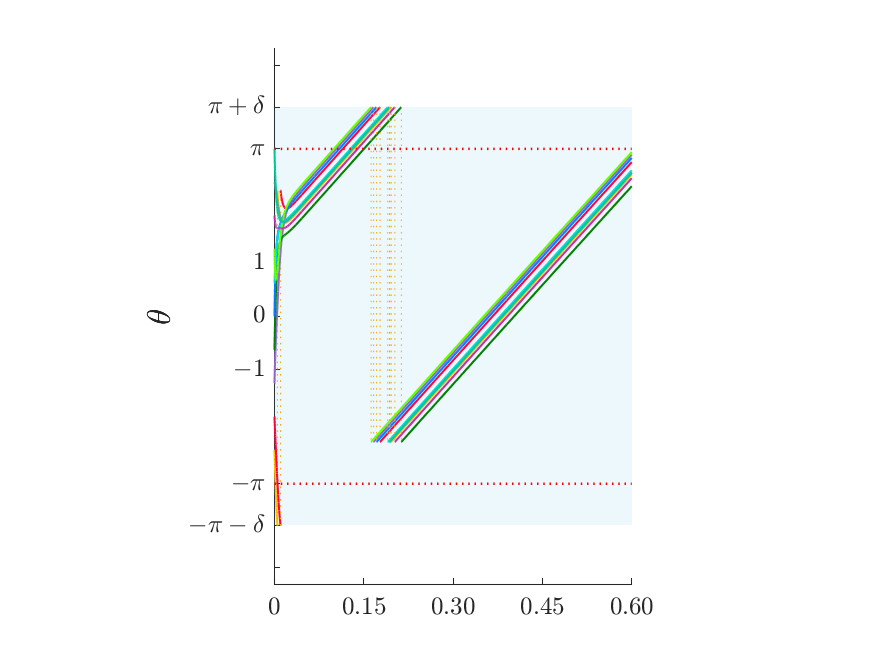}
                \includegraphics[trim={4.5cm 0.167cm 4.8cm 0.1cm},width=0.301365\textwidth,clip]{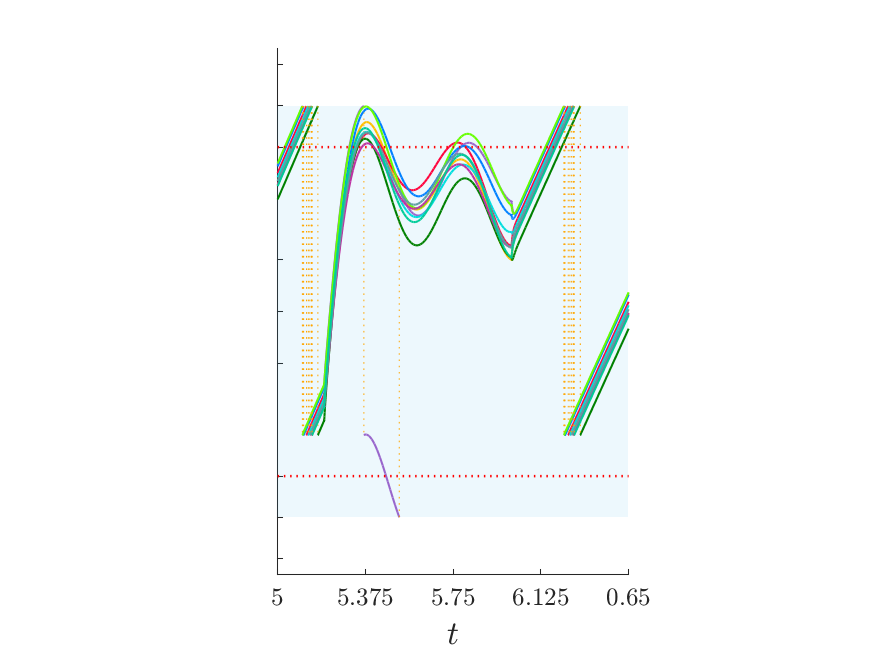} 
                \includegraphics[trim={4.5cm 0cm 5cm 0cm},width=0.28638\textwidth,clip]{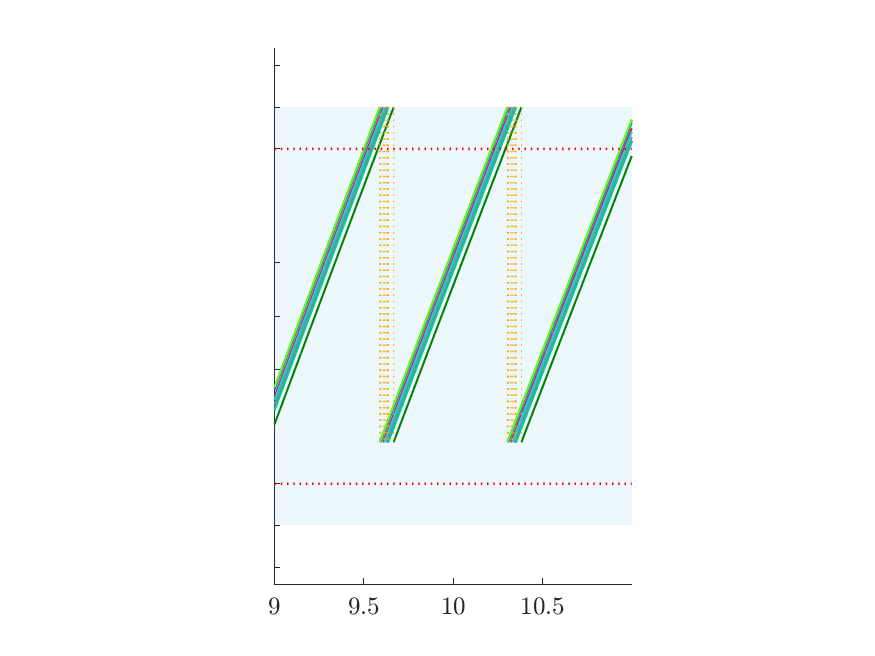}   \\
  	    \includegraphics[trim={1.6cm 0.1cm 2.6cm 0.1cm},width=0.85\textwidth,clip]{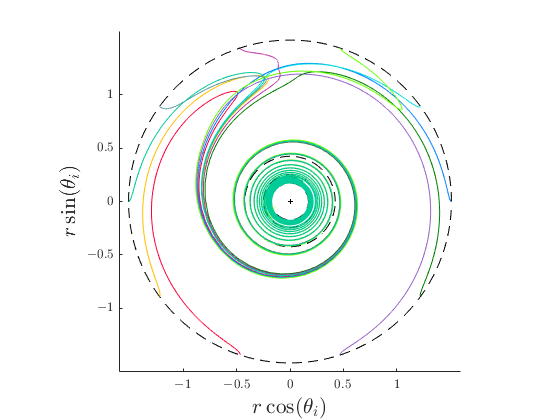}
\end{minipage}
\begin{minipage}[b]{0.32\textwidth}
\centering
\includegraphics[trim={0.1cm  2.0cm 1.1cm 1.9cm},width=1\textwidth,clip]{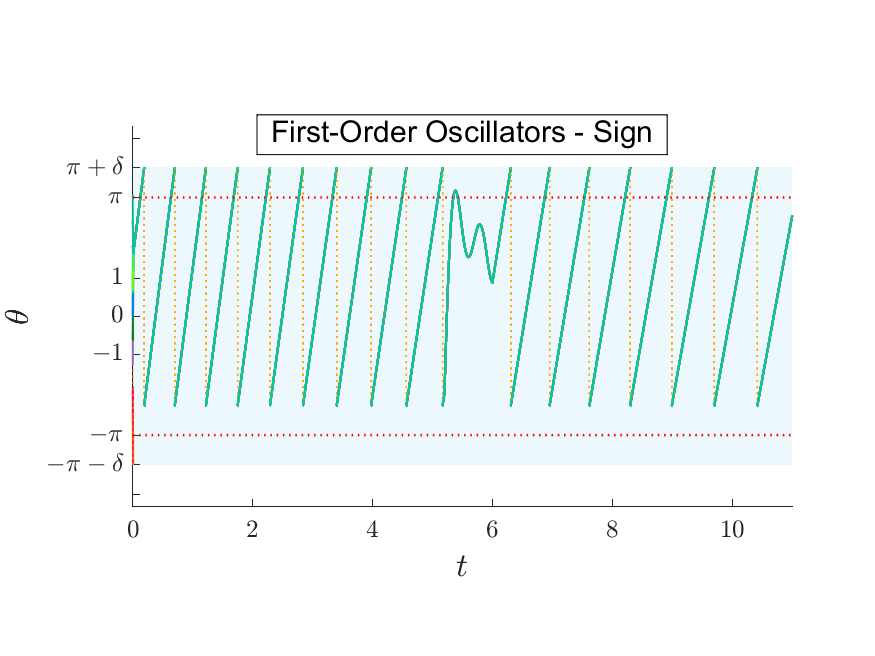}\\
                 \includegraphics[trim={2.5cm 0cm 5.5cm 0cm},width=0.363\textwidth,clip]{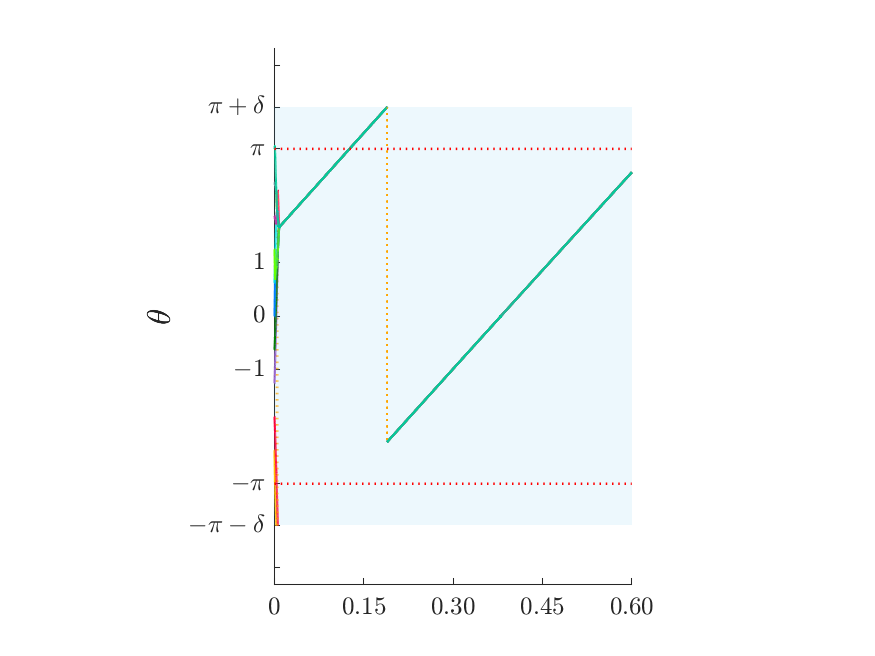}   
                \includegraphics[trim={4.5cm 0.167cm 4.8cm 0.1cm},width=0.301365\textwidth,clip]{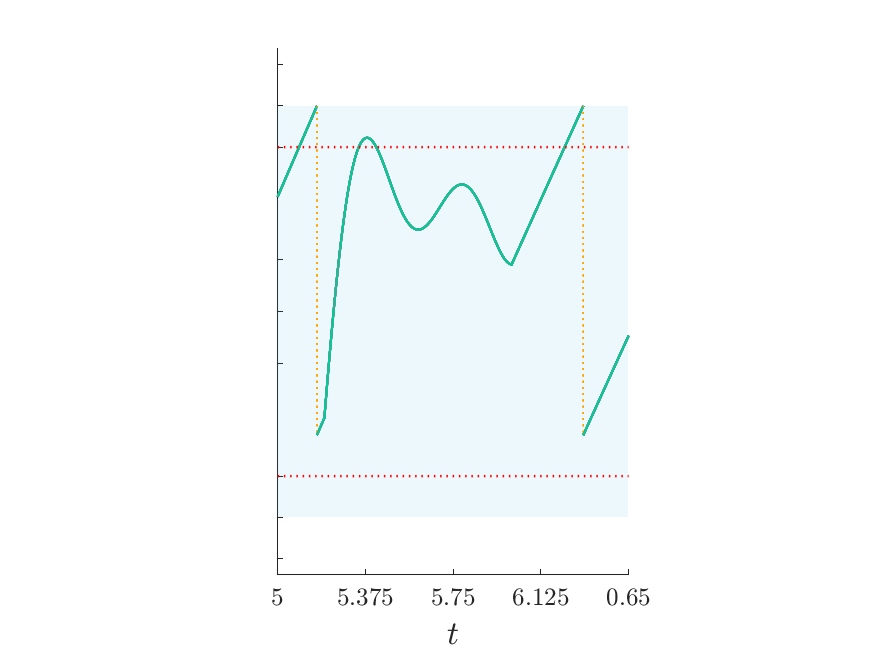} 
                \includegraphics[trim={4.5cm 0cm 5cm 0cm},width=0.28638\textwidth,clip]{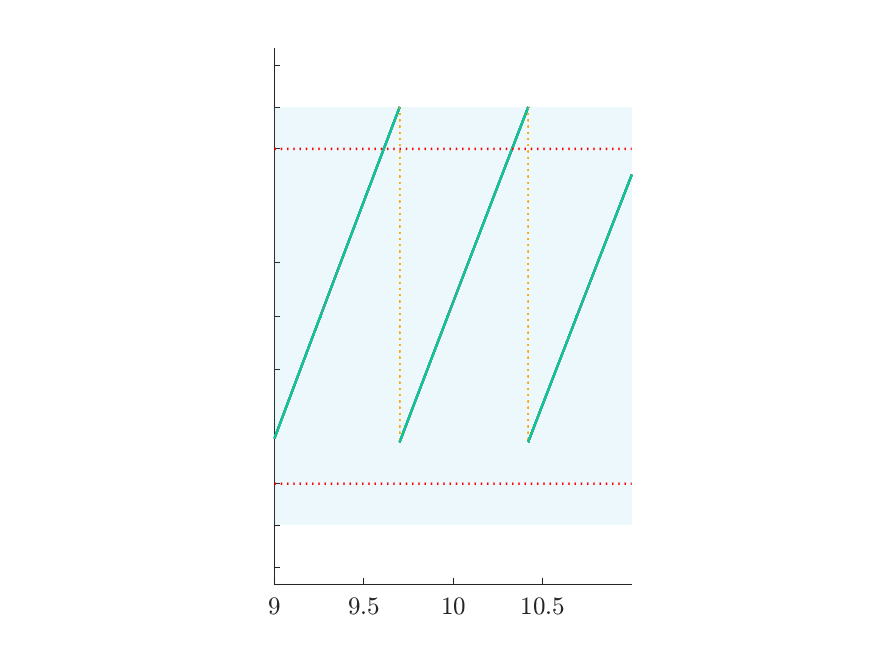} \\
  	    \includegraphics[trim={1.6cm 0.1cm 2.6cm 0.1cm},width=0.85\textwidth,clip]{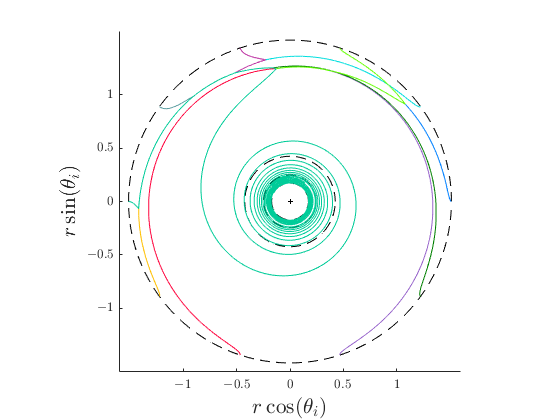}
       \end{minipage}
  \caption{(Top) Phase evolution for $\kappa= \frac{576 \pi}{10}$, $\delta=\frac{\pi}{4}$ and different selections of $\sigma$  and communication configurations. (Middle) Phase evolution in the time intervals $[0, 0.41]$, $[5,5.2]$ and $[9, 10.75]$. (Bottom) Evolution of the pair $(r\cos(\theta_{i}),r\sin(\theta_{i}))$, with  $r(\mathfrak{t})=(1.55\sqrt{\mathfrak{t}}+0.66)^{-1}$, showing radially the continuous-time evolution for the phases generated by our hybrid modification \eqref{eq:hybr_multi}. The black dashed lines are isotime ($0$ (outer), $3.66$,  $7.33$ and $11$ (inner) time units).}
  \label{fig:phases}
\end{figure*}
\begin{figure}[ht!]
	\begin{center}
		\, \includegraphics[trim={0.67cm 2.8cm 0 3.0cm},clip,width=0.79\columnwidth,clip]{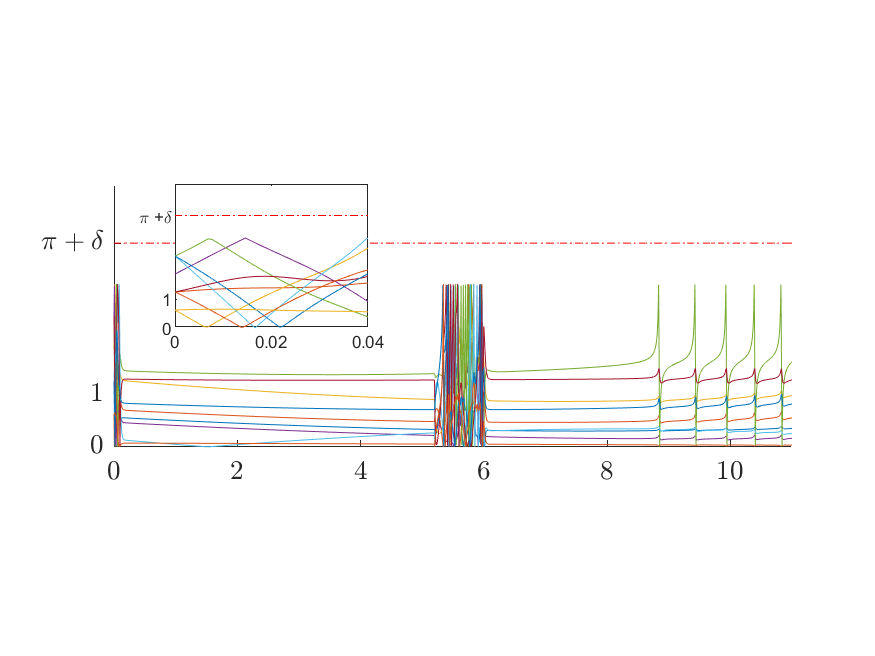}\vspace*{-0.1cm}\\
		\hspace*{-0.28cm} \includegraphics[trim={0 1.8cm 0 2.35cm},clip,width=0.825\columnwidth,clip]{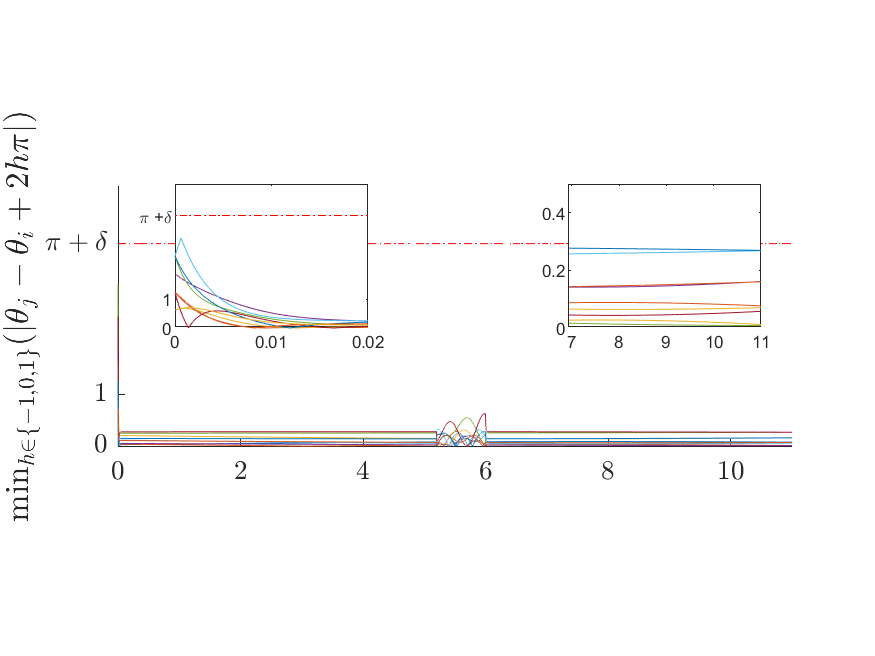}\vspace*{-0.1cm}\\
		\, \includegraphics[trim={0.67cm 2.3cm 0 3.0cm},clip,width=0.79\columnwidth,clip]{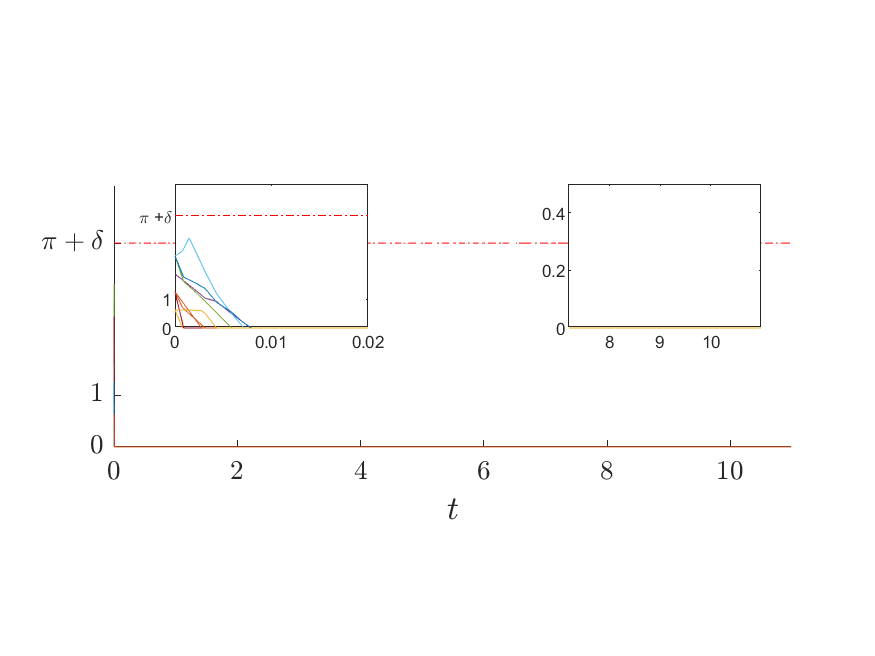}
	\end{center}
	\caption{Evolution of the phase errors for $\kappa=\frac{576 \pi}{10}$ and different selections of $\sigma$ and communication configurations (from top to bottom: no communication layer, ramp and sign functions).}
	\label{fig:errors}
\end{figure}

\section{Numerical illustration}
\label{sec:sims}

In this section, we apply our control scheme to globally, uniformly, synchronize the phases of $n=10$ strongly damped generators physically coupled over an all-to-all network
connection (represented by the blue edges in Figure~\ref{fig:network}) over a set $\mathcal{V}$ of nodes whose dynamics are approximated by nonuniform first-order Kuramoto oscillators given by \cite[eq. (2.8)]{dorfler2012synchronization}. 
This fully connected dynamics can be accurately described by the terms $\omega_i$'s in \eqref{eq:flow_phase_nagent_unidrected} with the following selection:
\begin{align}
\omega_i(\theta,\mathfrak t):= &\frac{1}{\zeta_{i}}\Bigg(\!\!\!\widetilde \omega_i \Big (1 + \frac{3}{10} \sin({\chi_i}\mathfrak{t}+ {\phi_i})\Big) + d_{i}(\mathfrak{t})-\widetilde\kappa_{ij} \!\!\!\!\! \sum_{j\in\mathcal{V} \setminus \{i\}} \!\!\! \sin(\theta_j - \theta_i + \phi_{ij})\Bigg),  \;
\forall i \in\mathcal{V},
\label{eq:wsims} 
\end{align}
for generic constant parameters ${\chi_i} \in \uni ([-1,1])$, $\widetilde \omega_{i} \in \uni([-5,5])$, ${\phi_i}\in \linebreak \uni([0,\atan(0.25)])$, and $\zeta_{i}\in \uni([20, 30]\frac{1}{120 \pi})$. The physical all-to-all coupling among the oscillators 
(blue edges in Figure~\ref{fig:network}) is modeled by the sine functions 
$\sin(\theta_j - \theta_i + \phi_{ij})$ of the angular mismatch between the oscillators offsetted by the constant angle $\phi_{ij}\in \uni([0,\atan(0.25)])$. Furthermore, each physical coupling is scaled by the gain $\widetilde \kappa_{ij}=\widetilde \kappa_{ji} \in \uni([0.7, 1.2])$. This allows fully embedding in our time-varying generalized natural frequencies $\omega_i$ of \eqref{eq:flow_phase_nagent_unidrected} the physical couplings of the oscillators. Each high-frequency disturbance $d_{i}:\real_{\geq 0} \rightarrow [0,5]$ is defined as $d_{i}(s)=0$ if $s\in [0,5.2] \cup [6.0,11]$ and $d_i(s)=5\sin(50 {\chi_i}s+ {\phi_i})$ if $s\in (5.2,6)$. Thus $\omega_i (\theta,\mathfrak{t})$ not only captures the time-varying natural frequency of the $i$-th oscillator but also the physical coupling actions and disturbances influencing its dynamics. The parameters have been selected as in (\cite{dorfler2012synchronization}) to model realistic, strongly damped, generators. On the other hand,  the synchronizing coupling actions are exchanged through our communication graph $\G_u = (\V, \E_u)$, whose edges are depicted in red in Figure~\ref{fig:network}. These ``cyber'' coupling actions are represented by the functions $\sigma$'s in \eqref{eq:flow_phase_nagent_unidrected}, whose design is performed according to our solution of Section \ref{sec:model_disc}.
Summarizing, the combination of the (blue) physical layer and the (red) ``cyber'' communication layer of Figure~\ref{fig:network} generates a cyber-physical system whose dynamics is represented by \eqref{eq:flownD_compact}, \eqref{eq:hybr_multi}, with $\boldsymbol{\omega}$ capturing the physical layer and $\boldsymbol{\sigma}$ capturing the hybrid feedback control action. We initialize the oscillators with $q(0,0)=\boldsymbol{0}_{9}$ and the initial phases are chosen in such a way that the oscillators are equally spaced on the unit circle. Finally, we select $\delta=\frac{\pi}{4}$. 

 The evolution of the phases, $\theta_i's$, and the angular errors between any two neighbours in $\mathcal{G}$, namely $\min_{h\in\{-1,0,1\}}(|\theta_j-\theta_i+2h \pi|)$,  are reported\footnote{The simulations have been carried out using the Matlab toolbox HyEQ (\cite{SanfeliceToolbox}).} in the top two rows of Figure~\ref{fig:phases} and in Figure~\ref{fig:errors}, for different selections of $\sigma$, and $\kappa=\frac{576 \pi}{10}$,  which ensure finite-time synchronization due to the bound reported in Lemma \ref{lem:geometricVdots_finitetime} in Section \ref{subsec:proofProp3}. When no communication layer is considered (left plots), the oscillators do not synchronize. 
 When the communication layer is implemented and $\sigma$ is given as the ramp function,  practical synchronization is achieved as established in Theorem~\ref{thm:practical_stability} and shown in Figures~\ref{fig:phases} and \ref{fig:errors}.
(Non-uniform) practical synchronization can also be achieved in the absence of a communication layer by selecting larger values for the $\widetilde\kappa_{ij}$'s. 
On the other hand, the sign function, which is discontinuous at $0$, also leads to a finite-time synchronization property in agreement with Theorem~\ref{thm:fixed-time}, see Figures~\ref{fig:phases} and \ref{fig:errors}. Furthermore, the set of plots in the bottom row of Figure~\ref{fig:phases} shows that each phase $\theta_i$ maps continuously the angular values identifying oscillator $i$ on the unit circle, in agreement with Section~\ref{subsec:flow_model_nagent} and Lemma~\ref{lem:jump_theta_prop}.

\begin{figure*}[t]
\begin{minipage}[b]{0.32\textwidth}
\centering
        \includegraphics[trim={0.0cm 1.0cm 0.8cm 1.4cm},width=1\textwidth,clip]{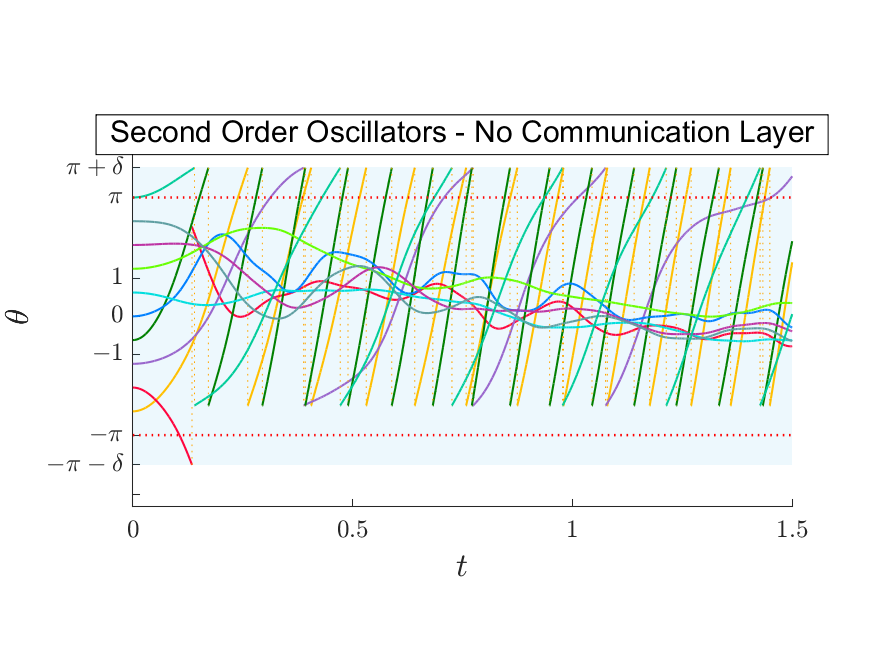}\\
        \includegraphics[trim={1.6cm 0.1cm 2.6cm 0.1cm},width=0.85\textwidth,clip]{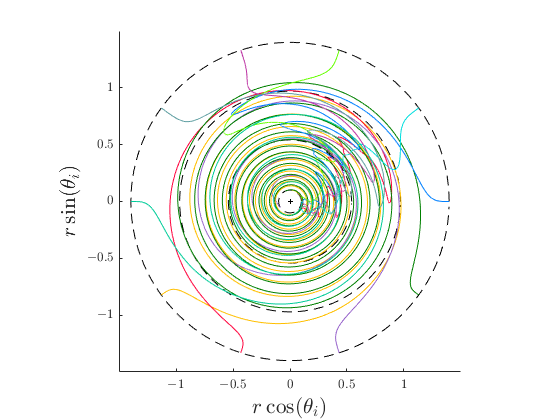} 
 \end{minipage}
 \begin{minipage}[b]{0.32\textwidth} 
 \centering
        \includegraphics[trim={0.0cm  1.0cm 0.8cm 1.4cm},width=1\textwidth,clip]{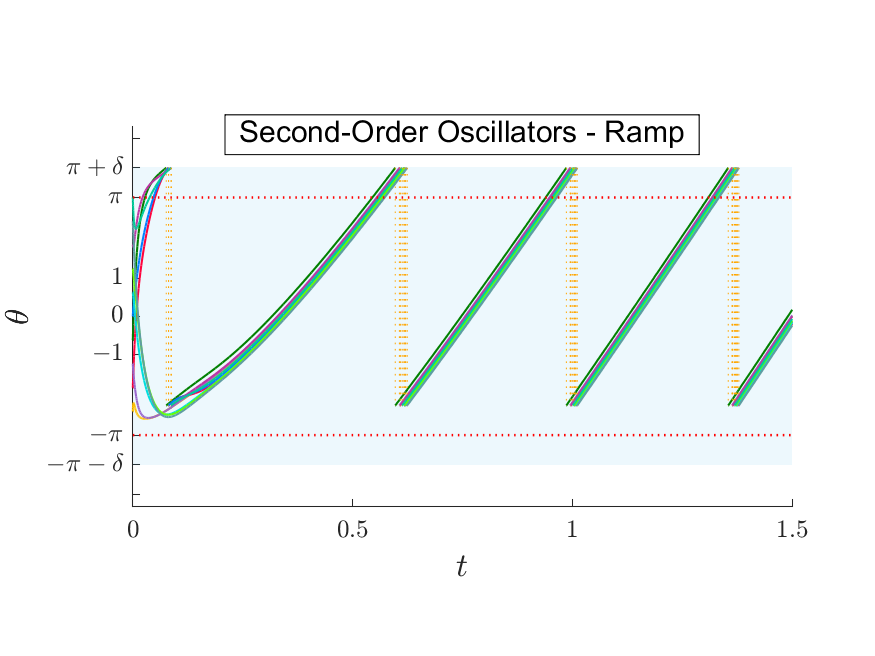} \\
  	    \includegraphics[trim={1.6cm 0.1cm 2.6cm 0.1cm},width=0.85\textwidth,clip]{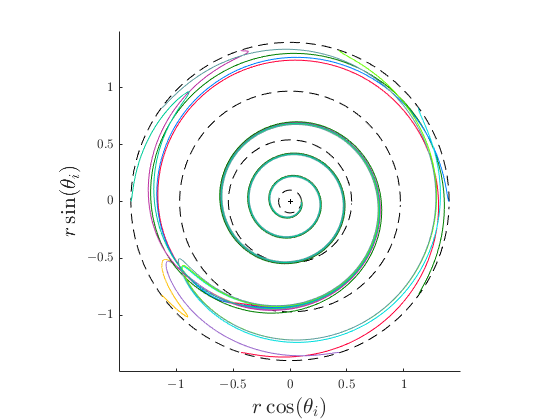}  
 \end{minipage}
 \begin{minipage}[b]{0.32\textwidth}
 \centering
        \includegraphics[trim={0.0cm  1.0cm 0.8cm 1.4cm},width=1\textwidth,clip]{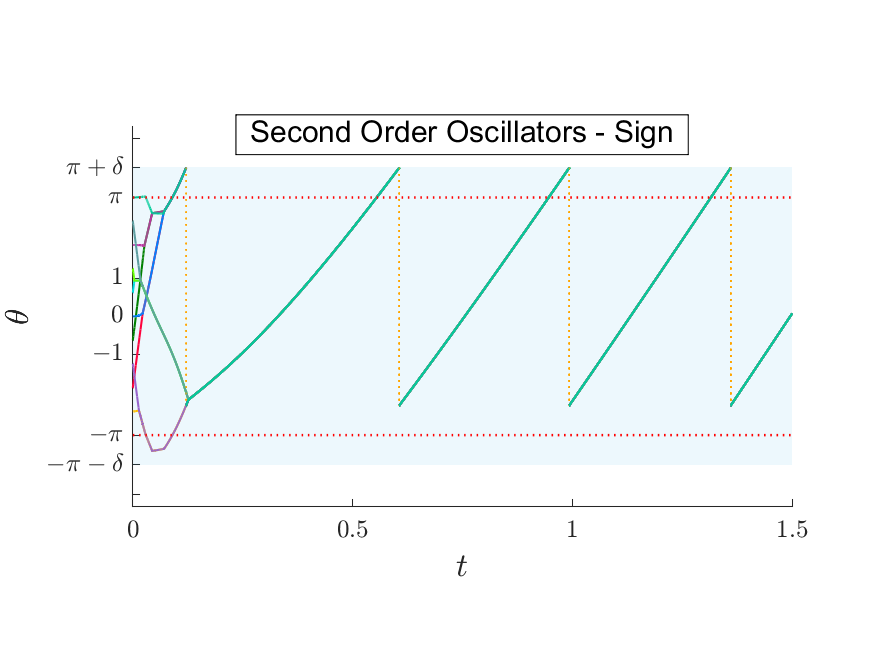}\\
  	    \includegraphics[trim={1.6cm 0.1cm 2.6cm 0.1cm},width=0.85\textwidth,clip]{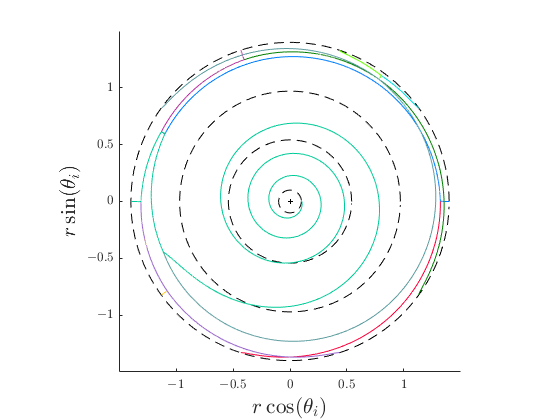}
 \end{minipage}
  \caption{(Top) Phase evolution associated of second-order oscillators for $\kappa=\frac{576 \pi}{10}$, $\delta=\frac{\pi}{4}$ and different selections of $\sigma$  and communication configurations. (Bottom) Evolution of the pair $(r\cos(\theta_{i}),r\sin(\theta_{i}))$, with  $r(\mathfrak{t})=-0.255\sqrt{\mathfrak{t}}+1.4$, showing radially the continuous-time evolution for the phases generated by our hybrid modification of second-order oscillators. The black dashed lines are isotime ($0$ (outer), $0.5$,  $1$ and $1.5$ (inner) time units).}
  \label{fig:phases2d}
\end{figure*}

The same exact hybrid controller dynamics is finally exploited in a more sophisticated context of non-strongly damped generators (rather than the strongly damped case, as considered above and in Figures~\ref{fig:phases} and~\ref{fig:errors}). Following (\cite{dorfler2012synchronization}), such behavior is modeled by 
a fully connected graph comprising the
second-order (rather than first-order) heterogeneous oscillators in \cite[eq. (2.3)]{dorfler2012synchronization}.
Once again, this physical interconnection is well represented by the blue edges in Figure~\ref{fig:network} and dynamics \eqref{eq:flow_phase_nagent_unidrected} with the following selection, generalizing \eqref{eq:wsims} to the dynamical context,
\begin{align}
\nonumber
\dot\omega_i(\theta,\mathfrak t)= &-\frac{\zeta_{i}}{m_{i}} \omega_i(\theta,\mathfrak t) + \frac{1}{m_{i}}\Bigg(\widetilde \omega_i \Big (1 + \frac{3}{10} \sin({\chi_i}\mathfrak{t}+ {\phi_i})\Big) \\ &+ {d_i}(\mathfrak{t}) -\widetilde\kappa_{ij} \sum_{j\in\mathcal{V} \setminus \{i\}} \sin(\theta_j - \theta_i + \phi_{ij})\Bigg), \; \forall i \in \mathcal{V}.  \label{eq:wsimsdot} 
\end{align}
This dynamically generalized selection uses the same parameters as in the previous set of simulations, with the addition of the constant mass parameters $m_i\in \uni([2, 12]\frac{1}{120 \pi})$ for each $i\in \V$, which is defined according to \cite[eq. 2.3]{dorfler2012synchronization}.
We apply our hybrid feedback control algorithm to this generalized scenario by augmenting once again the physical layer with a ``cyber'' communication layer represented by the red edges in Figure~\ref{fig:network}, inducing the 
stabilizing action of inputs $\boldsymbol{\sigma}$ in the hybrid interconnection \eqref{eq:flownD_compact}, \eqref{eq:hybr_multi}.
We initialize $q$ and $\theta$ as in the previous simulations and  $\dot \theta_i(0,0) \in \uni([-0.1, 0.1])$ for each $i\in \V$.  Finally, we still select $\delta=\frac{\pi}{4}$ in \eqref{eq:X}. The evolution of the phases is reported in Figure~\ref{fig:phases2d}, for different selections of $\sigma$, and $\kappa=\frac{576 \pi}{10}$. Similarly to what happens for the first-order oscillators of Figures~\ref{fig:phases} and~\ref{fig:errors}, when no communication layer is considered, the second-order oscillators do not synchronize. When the generators are equipped with the communication network and $\sigma$ is instead defined as the ramp function, practical synchronization is achieved, as predicted by Theorem~\ref{thm:practical_stability} and as shown in Figure~\ref{fig:phases2d}. On the other hand, considering the (discontinuous at $0$) sign function to generate the synchronizing hybrid coupling actions again leads to a finite-time synchronization property, thus confirming Theorem~\ref{thm:fixed-time}.

\section{Proof of Propositions~\ref{prop:PropV_practical} and \ref{prop:PropV_finite_time}}
\label{sec:proofs}
\subsection{Results on scalar non-pathological functions}
\label{subsec:nonpath}
The proofs of Propositions~\ref{prop:PropV_practical} and \ref{prop:PropV_finite_time}, which are instrumental for proving our main results of Sections~\ref{sec:as_stability} and \ref{sec:fixed-time}, require exploiting results from non-smooth analysis, because $V$ in \eqref{eq:lyapFunc_practical_V} is not differentiable everywhere. A further complication emerges from the fact that, since $\sigma$ may be discontinuous, the flow map in dynamics \eqref{eq:hybr_multi_reg_full} is outer semicontinuous, but not inner semicontinuous. The lack of inner semicontinuity prevents us from exploiting the ``almost everywhere'' conditions of (\cite{DellaRossaMCSS21}) and references therein. Instead, one could resort to conditions involving Clarke's generalized directional derivative and Clarke's generalized gradient, which can be defined as (see \cite[page 11]{Clarke90})
\begin{align}
  \partial V(x):=\text{co}\{\lim\limits_{i \to  \infty} \nabla  V(x_i): x_i \to x, \, x_i \notin \mathcal{Z}, \,  x_i \notin \Omega_u\},
\label{eq:clarke}
\end{align}
 where $\Omega_u$ is the set (of Lebesgue measure zero) where $V$ is not differentiable, and $\mathcal{Z}$ is any other set of Lebesgue measure zero. However, due to the peculiar dynamics considered here (resembling, for example, the undesirable conservativeness highlighted in \cite[Ex. 2.2]{DellaRossaPhD21}), Lyapunov decrease conditions based on Clarke's generalized gradient would be too conservative and impossible to prove. 
Due to the above motivation, in this section we prove Propositions~\ref{prop:PropV_practical} and \ref{prop:PropV_finite_time}  by exploiting the results of (\cite{Valadier89,BacCer03}), whose proof is also reported in \cite[Lemma 2.23]{DellaRossaPhD21}, establishing a link between the time derivative   $\frac{d}{d \mathfrak{t}}V(\phi(\mathfrak{t}))$ of a Lyapunov-like function $V$ evaluated along a generic solution $\phi$ of a continuous-time system, and the so-called set-valued Lie derivative (\cite{BacCer99})
\begin{align}
  \dot {\overline V}_F(x) := \{ a\in \real | \; \exists f \in F(x):  \langle v,f \rangle = a,\; \forall v\in \partial V(x)\},
  \label{eq:LieD}
\end{align}
with $\partial V$ defined in \eqref{eq:clarke}.
In the following, we characterize some features of the set-valued Lie derivative, useful for the next technical derivations.
\begin{lemma}
	\label{lem:LieDSingleton}
	Consider a set $S \subset \real^n$, $F:\dom F \subset \real^n \rightrightarrows \real^n$ with $S  \subset \dom(F)$ and a locally Lipschitz $V:\dom V \subseteq \real^n \to \real$ such that $S \subset \dom(V)$. Given any  function $\varphi:\real^n \times \real^n \to \real^n$ satisfying $\varphi(x,f)\in\partial V(x)$, $\forall x \in S$, $\forall f\in F(x)$, it holds that  
	\begin{equation}
		\label{eq:supLieD}
		\sup\dot {\overline V}_F(x) \leq\sup\limits_{\substack{f \in F(x)}} \langle \varphi(x,f), f \rangle.
	\end{equation} \null \hfill $\square$
\end{lemma}
\begin{proof}
 Consider any $x\in S$. In view of \eqref{eq:LieD}, and by the fact that $\varphi(x,f) \in \partial V(x)$ for all $f \in F(x)$, we have that $a=\langle v,f(x) \rangle \in \dot {\overline V}_F(x)$ implies $a=\langle \varphi(x,f), f \rangle$. Hence, we derive that $\dot {\overline V}_F(x) \subseteq \bigcup\limits_{f\in F(x)} \{\langle \varphi(x,f), f \rangle\}$ and thus $ \sup \dot {\overline V}_F(x) \leq \sup\limits_{f\in F(x)}\langle \varphi(x,f), f \rangle$ as to be proven.
\end{proof}

Whenever function $V$ is non-pathological (according to the definition given next), \cite[Lemma 2.23]{DellaRossaPhD21} ensures that $\frac{d}{d \mathfrak{t}}V(\phi(\mathfrak{t})) \in  \dot {\overline V}_F(\phi(\mathfrak{t}))$ for almost all $\mathfrak{t}$ in the domain of $\phi$. We provide below the definition of non-pathological functions and we prove that function $V$ in \eqref{eq:lyapFunc_practical_V} enjoys that property. Our result below can be seen as a corollary of the fact that piecewise $C^1$ functions (in the sense of (\cite{Scholtes})) are non-pathological. This result has been recently published in \cite[Lemma 4]{DellaRossa20}. The scalar case is a consequence of Proposition~5 in (\cite{Valadier89}). An alternative proof is reported here.
\begin{definition} (\cite{Valadier89})
\label{def:non-path}
A locally Lipschitz function $W:\dom W\subseteq \real \to \real$ is \emph{non-pathological} if, given any absolutely continuous function $\phi : \real_{\geq 0} \to \dom W$, we have that for almost every $\mathfrak{t}\in\real_{\geq 0}$ there exists $a_\mathfrak{t}\in\real$ satisfying 
\begin{equation}
	\label{eq:non-path eq}
	\langle w, \dot \phi (\mathfrak{t}) \rangle = a_\mathfrak{t}, \quad \forall w \in \partial W(\phi (\mathfrak{t})).
\end{equation} 
In other words, for almost every $\mathfrak{t} \in \real_{\geq 0}$, $\partial W(\phi(\mathfrak{t}))$ is a subset of an affine subspace orthogonal to $\dot \phi (\mathfrak{t})$. \null \hfill $\square$
\end{definition}

\begin{proposition} 
 \label{prop:nonpathological}
  Any function $W:\dom W \subseteq \real \to \real$ that is piecewise continuously differentiable is non-pathological. \null \hfill $\square$
\end{proposition}
For this paper, the interest of Proposition~\ref{prop:nonpathological} stands in the fact that it implies that function $V$ in \eqref{eq:lyapFunc_practical_V} is non-pathological (\cite{Valadier89,BacCer99}), being the sum of piecewise continuously differentiable scalar functions.  
\begin{remark}
An alternative proof of Proposition~\ref{prop:nonpathological}, might be to first establish that piecewise $C^1$ functions from $\real \to \real$ can be represented as a $\max$-$\min$ of $C^1$ functions from $\real \to \real$ (a similar result has been proven, for example, in \cite[Thm.~1 and Prop.~1]{Xu16} with reference to piecewise affine functions), and then obtain Proposition~\ref{prop:nonpathological} as a corollary of \cite[Lemma 2.20]{DellaRossaPhD21,Valadier89,BacCer99}, which establish non-pathological properties of $\max$-$\min$  functions.  
\end{remark}
\begin{proofname}{Proof of Proposition~\ref{prop:nonpathological}:}
Let $W:\dom W \subseteq \real \to \real$ be piecewise continuously differentiable. Let $\phi: \real_{\geq 0} \to \dom W$   be an absolutely continuous scalar function  and suppose it is differentiable at  $\overline{\mathfrak{t}} \in \real_{\geq 0}$. We split the analysis in three cases.
 \begin{enumerate}[label=\alph*), leftmargin=*, ref=\alph*]  
 	\item $\dot \phi(\overline{\mathfrak{t}})=0$. Then, for all  $w \in \partial W(\phi (\overline{\mathfrak{t}}))$, $\langle w, \dot \phi(\overline{\mathfrak{t}}) \rangle=\langle w, 0 \rangle=0.$ Thus, \eqref{eq:non-path eq} is satisfied with $a_\mathfrak{t}=0$.
 	\item $\dot \phi(\overline{\mathfrak{t}})>0$. If $W$ is continuously differentiable at $\phi(\overline{\mathfrak{t}})$ then, $\partial W(\phi(\overline{\mathfrak{t}}))= \linebreak \nabla W(\phi(\overline{\mathfrak{t}}))$ and \eqref{eq:non-path eq} holds. Consider now the case where $W$ is not differentiable at $\overline{\mathfrak{t}}$. We recall that, in view of the fact that $W$ is piecewise continuously differentiable, $\partial W(\phi(\mathfrak{t}))=\nabla W(\phi(\mathfrak{t}))$ for any $\mathfrak{t}$ in a sufficiently small neighbourhood of $\overline{\mathfrak{t}}$.  From the absolute continuity of $\phi$, there exists $\varepsilon>0$ such that $\lim\limits_{h \to 0}\frac{\phi(\overline{\mathfrak{t}}+\mathfrak{h})-\phi(\overline{\mathfrak{t}})}{\mathfrak{h}}\geq\varepsilon>0$. Hence, there exists $\rho_1 \in \real_{> 0}$ such that for any $\rho \in (0, \rho_1]$, we have  $\frac{\phi(\overline{\mathfrak{t}}+\rho)-\phi(\overline{\mathfrak{t}})}{\rho}\geq\frac{\varepsilon}{2}$ and thus  $\phi(\overline{\mathfrak{t}}+\rho) \geq \frac{\rho\varepsilon}{2} + \phi(\overline{\mathfrak{t}}) >\phi(\overline{\mathfrak{t}})$. With a similar reasoning, there exist $\rho_2 \in \real_{\geq 0}$ such that $\rho_2>0$  implies $\phi(\overline{\mathfrak{t}}-\rho)  < \phi(\overline{\mathfrak{t}})$ for any $\rho\in(0,\rho_2]$. Hence, there exists a neighbourhood of $\overline{\mathfrak{t}}$ contained in  $[\overline{\mathfrak{t}}- \rho_2,\overline{\mathfrak{t}}+ \rho_1]$, for which $\partial W(\phi(\cdot))$ is defined and coincides with $\nabla W(\phi(\mathfrak{t}))$. Therefore, for almost every $\tilde{\mathfrak{t}} \in \mathcal{T} \subseteq [\overline{\mathfrak{t}}- \rho_2,\overline{\mathfrak{t}}+ \rho_1]$ there exists $a_\mathfrak{t}\in\real$ such that \eqref{eq:non-path eq} is satisfied.
 	\item $\dot \phi(\overline{\mathfrak{t}})<0$. This case is identical to the previous one by changing all the signs, therefore implying that for almost every $\mathfrak{t}$ there exists $a_\mathfrak{t}\in\real$  in a compact neighbourhood of  $\overline{\mathfrak{t}}$ such that \eqref{eq:non-path eq} is satisfied.    
 \end{enumerate}
 Since $\phi$ is absolutely continuous, then the set  $\Phi$ where it is not differentiable is of Lebesgue measure zero. We conclude that \eqref{eq:non-path eq} is satisfied for almost all $\mathfrak{t} \in \real_{\geq 0}$ because we have arbitrarily selected  $\overline{\mathfrak{t}} \in \real_{\geq 0}\setminus\Phi$, as to be proven.
\end{proofname}
\subsection{Proof of Proposition~\ref{prop:PropV_practical}}
\label{subsec:proofProp2}
The following lemma 
establishes geometric properties of $V$ that are used, together with
Lemma~\ref{lem:sandwich_etal} and
Proposition~\ref{prop:nonpathological}, to show the trajectory-based results of Proposition~\ref{prop:PropV_practical}.
\begin{lemma} \label{lem:geometricVdots}
    Consider system \eqref{eq:hybr_multi_reg_full}, function $V$ in \eqref{eq:lyapFunc_practical_V}-\eqref{eq:lyapFunc_ij_practical} and $c$ as in Theorem~\ref{thm:practical_stability}. There exists $\alpha_3 \in \K_\infty$, independent of $\overline \omega$ in \eqref{eq:pars} and of $\kappa$, such that
  \begin{subequations} 
  \begin{align}
 &\sup \dot {\overline V}_F(x) \leq - \kappa\underline \lambda\alpha_3(V(x)) +c\overline\omega, &&\quad \forall x\in C,
 \label{eq:flowCond_practical}\\
&\Delta V(x):= \sup \limits_{g \in G(x)}V(g) - V(x) \leq 0, &&\quad \forall x \in D.
\label{eq:jumpCond_practical}
\end{align}
\end{subequations}\null \hfill $\square$
\end{lemma}
\begin{proof}
  We prove the two equations one by one.
  
  \emph{Proof of \eqref{eq:flowCond_practical}.}
  For each $x\in C$ and each $f\in F(x)$, there exist $\widehat{\boldsymbol{\omega}} \in \widehat \Omega$ and $\boldsymbol{\sigma}_{\!f} \in \widehat{\Sigma}(x)$ such that $f=(\widehat{\boldsymbol{\omega}}- B \kappa \boldsymbol{\sigma}_{\!f},0) \in F(x)$. Define $\varphi(x,f)$ in Lemma~\ref{lem:LieDSingleton} as $\varphi(x,f):=(B\boldsymbol{\sigma}_{\!f},2\pi\boldsymbol{\sigma}_{\!f})$. From \eqref{eq:LieD} in Lemma~\ref{lem:LieDSingleton}, we have
  \begin{equation}
  	\label{eq:LieVstarting}
  	\sup\dot{\overline{V}}_F(x) \leq \sup\limits_{\substack{\widehat{\boldsymbol{\omega}} \in \widehat{\Omega}, \\ \boldsymbol{\sigma}_{\!f} \in \widehat{\Sigma}(x)}} (-\kappa \boldsymbol{\sigma}_{\!f}^\top B ^\top B \boldsymbol{\sigma}_{\!f} + \boldsymbol{\sigma}_{\!f}^\top B^\top \widehat{\boldsymbol{\omega}}).
  \end{equation} 
Thus, in view of \eqref{eq:pars} and Lemma~\ref{lem:B_prop}, \eqref{eq:LieVstarting} yields 
 \begin{align}
	\label{eq:LieVmid}
	\begin{split}
		\sup\dot{\overline{V}}_F(x)&\leq \sup\limits_{\substack{\boldsymbol{\sigma}_{\!f} \in \widehat{\Sigma}(x)}} (-\kappa \boldsymbol{\sigma}_{\!f}^\top B ^\top B \boldsymbol{\sigma}_{\!f}) + c \overline{\omega} \\ &\leq  \sup\limits_{\substack{\boldsymbol{\sigma}_{\!f} \in \widehat{\Sigma}(x)}} -\kappa \underline\lambda |\boldsymbol{\sigma}_{\!f}|^2 + c \overline{\omega}.
	\end{split}
\end{align}
Moreover, in view of \eqref{eq:sandwich_V} and Lemmas~\ref{lem:sandwich_etal} and \ref{lem:lowerbound_sigma}, we have that, for any  $\boldsymbol{\sigma}_{\!f} \in \widehat{\Sigma}(x)$,
\begin{equation}
  \label{eq:upper_bound_sand_V}
  \eta \circ \alpha_2^{-1}(V(x)) \leq \eta(|x|_{\A}) \leq |\boldsymbol{\sigma}_{\!f}|^2.
\end{equation}
  By defining  $\alpha_3:=\eta\circ\alpha_2^{-1}\in\K_\infty$, \eqref{eq:flowCond_practical} stems from \eqref{eq:LieVmid} and \eqref{eq:upper_bound_sand_V}.
  
  \emph{Proof of \eqref{eq:jumpCond_practical}.}
  We split the analysis in two cases.
  
  First, let $i \in \V$, $x \in D_i$ and $x^+ = g_i(x)$ as in Lemma~\ref{lem:jump_theta_prop}.
  In view of the equality in \eqref{eq:jump_theta_prop} and the definition of $V_{ij}$ in \eqref{eq:lyapFunc_ij_practical}, we have
  $
    V_{ij}(x^+) = \int_{0}^{\theta_j^+ - \theta_i^+ + 2 q_{ij}^+ \pi} \sigma(\sat_{\pi + \delta}(s)) ds 
    = \int_{0}^{\theta_j - \theta_i + 2 q_{ij} \pi} \sigma(\sat_{\pi + \delta}(s)) ds = V_{ij}(x),
  $
  and thus
  $V(x^+) =\!\!\! \sum\nolimits\limits_{(i,j) \in \E} V_{ij}(x^+) =\!\!\! \sum\nolimits\limits_{(i,j) \in \E} V_{ij}(x) = V(x).$
  
  Second, let $(i,j)\in\E$, $x \in D_{ij}$ and $x^+ \in G_{ij}(x)$ as in Lemma~\ref{lem:jump_k_prop}. In view of item~\ref{prop:sigma_symm}) of Property \ref{prop:sigma},
  \begin{align}
    \label{eq:lyap_symm_practical}
    V_{ij}(x) = \int_0^{|\theta_j - \theta_i + 2 q_{ij} \pi|} \sigma(\sat_{\pi+\delta}(s)) ds.
  \end{align}
  On the other hand, in view of \eqref{eq:Dij_multi}, \eqref{eq:G_ij} and Lemma~\ref{lem:jump_k_prop},
 $
    |\theta_j^+ - \theta_i^+ +2 q_{ij}^+ \pi| = |\theta_j - \theta_i + 2 q_{ij}^+ \pi|
    < \pi + \delta \leq |\theta_j - \theta_i + 2 q_{ij} \pi|,
  $
  because $x\in D_{ij}$.
  Consequently, in view of \eqref{eq:Dij_multi} and item~\ref{prop:sigma_sect}) of Property~\ref{prop:sigma}
    $V_{ij}(x^+) = \int_0^{|\theta_j^+ - \theta_i^+ + 2 q_{ij}^+ \pi|}  \sigma(\sat_{\pi+\delta}(s)) ds 
    < \int_0^{|\theta_j - \theta_i + 2 q_{ij} \pi|}  \sigma(\sat_{\pi+\delta}(s)) ds = V_{ij}(x)$.

    On the other hand, from the definition of  $G_{ij}$ in \eqref{eq:G_ij}, $V_{uv}(x^+) = V_{uv}(x)$ for any $(u,v)\neq (i,j) \in \E$.
  Therefore $V(x^+) - V(x) = V_{ij}(x^+) - V_{ij}(x) < 0$, since the arguments of all the other elements of the summation in \eqref{eq:lyapFunc_practical_V} do not change.
\end{proof}

Based on Lemma~\ref{lem:geometricVdots}, we can now prove Proposition~\ref{prop:PropV_practical}.

\begin{proofname}{Proof of Proposition~\ref{prop:PropV_practical}:}
Item~(ii) of Proposition~\ref{prop:PropV_practical} is a direct consequence of \eqref{eq:jumpCond_practical} in Lemma~\ref{lem:geometricVdots}. To prove item~(i) of Proposition~\ref{prop:PropV_practical} we exploit the fact that, in view of \cite[Lemma 2.23]{DellaRossaPhD21} and $V$ being non pathological, for each solution $x$ to \eqref{eq:hybr_multi_reg_full}, for all $\mathfrak{j} \in \{0, \ldots, J\}$ and almost all $\mathfrak{t} \in [\mathfrak{t}_\mathfrak{j}, \mathfrak{t}_{\mathfrak{j}+1}]$ in $\dom x$, 
$\frac{d V(x(\mathfrak{t},\mathfrak{j}))}{d\mathfrak{t}} \in \dot{\overline{V}}_F(x(\mathfrak{t},\mathfrak{j}))$. Hence, as a consequence, $
\frac{d}{d\mathfrak{t}} V(x(\mathfrak{t},\mathfrak{j})) \leq - \kappa\underline \lambda\alpha_3(V(x(\mathfrak{t},\mathfrak{j}))) + c\overline\omega,$
 thus concluding the proof.
\end{proofname}
 \subsection{Proof of Proposition~\ref{prop:PropV_finite_time}}
 \label{subsec:proofProp3}
Paralleling Section~\ref{subsec:proofProp2},  we
establish via the next lemma geometric properties of $V$ that can be used, together with
Lemma~\ref{lem:sandwich_etal} and
Proposition~\ref{prop:nonpathological}, to show the trajectory-based result of Proposition~\ref{prop:PropV_finite_time}.
\begin{lemma} \label{lem:geometricVdots_finitetime}
 If $\sigma$ is discontinuous at the origin, then there exist $\mu \in \real_{>0}$ independent of $\overline \omega$ in \eqref{eq:pars} and $\kappa^\star>0$ such that for each $\kappa> \kappa^\star$
	\begin{subequations} 
		\begin{align}
			&\sup \dot {\overline V}_F(x) \leq -\frac{1}{2}\kappa \underline\lambda \mu^2, &&    \forall x\in C \setminus \A,
			\label{eq:flowCond_finite}\\
			&\Delta V(x):= \!\! \sup \limits_{g \in G(x)} \! V(g) - V(x) \leq 0, &&   \forall x \in D. 
			\label{eq:jumpCond_finite}
		\end{align}
	\end{subequations} \null \hfill $\square$
\end{lemma}
\begin{proof}
	We only prove \eqref{eq:flowCond_finite}, as \eqref{eq:jumpCond_finite} follows  from  the same arguments as those proving \eqref{eq:jumpCond_practical} in Proposition~\ref{prop:PropV_practical}. For each $x \in C \setminus \A$ and each $f\in F(x)$, we may proceed as in \eqref{eq:LieVstarting} and then exploit from  Lemma~\ref{lem:disc_sigma_lower_b} that, 
	\begin{align}
			\label{eq:LieVmid_fix_bound}
			\begin{split}
				\sup\dot{\overline{V}}_F(x)&\leq -\kappa \underline\lambda \mu^2 + c \overline{\omega}.
			\end{split}
	\end{align}
		By selecting $\kappa \geq \kappa^\star= \frac{2c\overline{\omega}}{\underline\lambda \mu^2}$, \eqref{eq:LieVmid_fix_bound} yields
$
				\sup\dot{\overline{V}}_F(x)\leq -\frac{1}{2}\kappa \underline\lambda \mu^2,
$
	which proves \eqref{eq:flowCond_finite}.
\end{proof}

Based on Lemma~\ref{lem:geometricVdots_finitetime}, we are now ready to prove Proposition~\ref{prop:PropV_finite_time}.

\begin{proofname}{Proof of Proposition~\ref{prop:PropV_finite_time}:}
Item~(ii) of Proposition~\ref{prop:PropV_finite_time} is a direct consequence of \eqref{eq:jumpCond_finite} in Lemma~\ref{lem:geometricVdots_finitetime}. In view of \cite[Lemma 2.23]{DellaRossaPhD21} and $V$ being non pathological,  for each solution $x$ to \eqref{eq:hybr_multi_reg_full}, for all $\mathfrak{j} \in \{0, \ldots, J\}$ and almost all $\mathfrak{t} \in [\mathfrak{t}_\mathfrak{j}, \mathfrak{t}_{\mathfrak{j}+1}]$ in $\dom x$, 
$\frac{d V(x(\mathfrak{t},\mathfrak{j}))}{d\mathfrak{t}} \in \dot{\overline{V}}_F(x(\mathfrak{t},\mathfrak{j}))$. Hence, as a consequence, 
$
		\frac{d}{d\mathfrak{t}} V(x(\mathfrak{t},\mathfrak{j})) \leq-\frac{1}{2}\kappa \underline\lambda \mu^2,
$
whenever $x(\mathfrak{t},\mathfrak{j}) \notin \A$, thus proving item~(i) of Proposition~\ref{prop:PropV_finite_time} which concludes the proof. 
\end{proofname}

\section{Conclusions}
\label{sec:concs}
We presented a cyber-physical hybrid model of leaderless heterogeneous first-order oscillators, where global uniform asymptotic and/or finite-time synchronization is obtained in a distributed way via hybrid coupling. More specifically we establish that the synchronization set for the proposed model enjoys uniform asymptotic practical stability property. Thanks to the mild requirements on the coupling function, the stability result was strengthened to a prescribed finite-time property when the coupling function is discontinuous at the origin. Finally, we proved a useful statement on scalar non-pathological functions, exploited here for the non-smooth Lyapunov analysis in our main theorems. We believe that this work demonstrates the potential of hybrid systems theoretical tools to overcome the fundamental limitations of continuous-time systems.

Future extensions of this work include addressing graphs with cycles (not trees) and investigating the case with leaders as done for a second-order Kuramoto model in (\cite{BossoLeaderKura}). Additional challenges may include studying the converse problem of globally de-synchronizing the network (\cite{DesyncKura}) via hybrid approaches. 

\bigskip

\noindent
{\bf Acknowledgement}.
The authors would like to thank Matteo Della Rossa and \linebreak Francesca Maria Ceragioli for useful discussions about the results of Section~8.1, Elena Panteley for useful discussions and suggestions given during the preparation of the manuscript, and the reviewers of the submission, which allowed us, through their comments, to improve the overall quality of the work.

\appendix
\section{Appendix} 
\begin{proofname}{Proof of Lemma~\ref{lem:jump_k_prop}:}
\label{proof:jump_k_prop}
Let $(i,j) \in \mathcal{E}$, $x \in D_{ij}$ and $x^+$ satisfies \eqref{eq:G_ij}.
Then, $\theta^+ = \theta$ while $q^+ \in \{ -1, 0, 1 \}^m$ in view of \eqref{eq:jump_rule_theta_ij}.
Thus, $x^+ \in X$ and the first part of the statement is proved. Let $\Delta \theta_{ij} := \theta_j - \theta_i$, so that $\theta_j - \theta_i + 2 q_{ij} \pi = \Delta \theta_{ij} + 2 q_{ij} \pi$. Since $x^+ \in X$, by definition of $X$ in \eqref{eq:X}, $|\Delta \theta_{ij}| < 2 \pi + 2 \delta$.

	We now prove that $|\Delta \theta_{ij}^+ + 2q_{ij}^+ \pi|=|\Delta \theta_{ij} + 2q_{ij}^+ \pi|< \pi + \delta$ by exploiting the fact that $q_{ij}^+=h^* \in \argmin\nolimits\limits_{h \in \{ -1, 0, 1 \}} |\Delta \theta_{ij} + 2 h \pi|$ according to \eqref{eq:G_IJ}, and splitting the analysis in five cases.
	\begin{enumerate}[label=\alph*), leftmargin=*, ref=\alph*]
		\item $\Delta \theta_{ij} \in (-\pi,\pi)$.
		Then, the minimizer is $h^* = 0$ and $|\Delta \theta_{ij} + 2 h^* \pi| < \pi < \pi + \delta.$
		\item $\Delta \theta_{ij} = \pi$.
		Then, the minimizer is $h^* \in \{-1,0\} $ and $|\Delta \theta_{ij} + 2 h^* \pi| \leq \pi < \pi + \delta.$
		\item $\Delta \theta_{ij}=-\pi$.
		This case is identical to the previous one by changing all the signs, therefore  $h^* \in \{0,1\} $  and $|\Delta \theta_{ij} + 2 h^* \pi| \leq \pi < \pi + \delta.$
		\item $\Delta \theta_{ij} \in (\pi, 2 \pi + 2 \delta]$.
		Then, the minimizer is $h^* = -1$ and $\Delta \theta_{ij} + 2 h^* \pi \in (-\pi, 2 \delta]$, which implies	$|\Delta \theta_{ij} + 2 h^* \pi| < \pi + \delta$, since $\max (2 \delta, \pi) < \pi + \delta$.
		\item $\Delta \theta_{ij} \in [-2 \pi - 2 \delta, - \pi)$.
		In this case, the minimizer is $h^* = 1$ and $\Delta \theta_{ij} + 2 h^* \pi \in [-2 \delta, - \pi)$, which implies	$|\Delta \theta_{ij} + 2 h^* \pi| < \pi + \delta$, since $2 \delta < \pi + \delta$.
	\end{enumerate}
	 Hence, we obtain, in view of all the previous cases, $|\theta_j^+ - \theta_i^+ +2 q_{ij}^+ \pi| \leq \max (2 \delta, \pi) < \pi + \delta$ thus concluding the proof since we have arbitrarily selected  $(i,j)\in \E$. 
\end{proofname}	 

\begin{proofname}{Proof of Lemma~\ref{lem:jump_theta_prop}:}
\label{proof:jump_theta_prop}
Let $i \in \mathcal{V}$, $x \in D_i$ and $x^+ = g_i(x)$. Let $(u,v) \in \mathcal{E}$, in view of \eqref{eq:jump_rule_theta}, if $u \neq i$ and $v \neq i$ the first equality in \eqref{eq:jump_theta_prop} trivially holds.
If $u = i$, then we have
$
	\theta_v^+ - \theta_u^+ + 2 q_{uv}^+ \pi 
	= \theta_v - \theta_u + \sgn(\theta_u)2 \pi + 2 (q_{uv} - \sgn(\theta_u)) \pi 
	= \theta_v - \theta_u + 2 q_{uv} \pi.
$
Similarly, we obtain for $v = i$ that
$
	\theta_v^+ - \theta_u^+ + 2 q_{uv}^+ \pi = \theta_v - \theta_u + 2 q_{uv} \pi,
$
thus proving the first equality in \eqref{eq:jump_theta_prop}.
On the other hand, in view of \eqref{eq:g_theta}, $|\theta_i^+| = \pi - \delta < \pi + \delta$.
Thus, all the elements of \eqref{eq:jump_theta_prop} are proved.

Let us now prove that $x^+ \in X$. In particular, we need to make sure that $q^+ \in \{ -1, 0, 1 \}^m$. For any $j \neq i \in \V$ we have that $\theta_j^+ = \theta_j$ and $|\theta_i^+| = \pi - \delta$ in view of \eqref{eq:g_theta}.
	Moreover, if $j$ is such that $(i,j) \in \E$, then from \eqref{eq:Dij_multi}, \eqref{eq:Di_multi} we prove next that $q_{ij} \neq -\sgn(\theta_i)$.
	Indeed, if $q_{ij} = -\sgn(\theta_i)$ then we would have
$
		|\theta_j - \theta_i + 2 q_{ij} \pi| = \big|\theta_j - \sgn(\theta_i)|\theta_i| - \sgn(\theta_i) 2 \pi \big| 
		= \big|\theta_j - \sgn(\theta_i)(3 \pi + \delta) \big| \geq 2 \pi > \pi + \delta,
$
	meaning that $x \in \mbox{int}(D_{ij})$ and consequently $x \not \in D_{i}$.
	Thus, $q_{ij}\neq - \sign(\theta_i)$ and, in view of \eqref{eq:g_k}, we obtain $q_{ij}^+ \in \{ 0, -\sgn(\theta_i) \}$.
	With a similar reasoning, we conclude that if $j$ is such that $(j,i) \in \E$ then we must have $q_{ji} \neq \sgn(\theta_i)$, implying that $q_{ji}^+ \in \{ 0, \sgn(\theta_i) \}$ in view of \eqref{eq:g_k}.
	Hence, $x^+ \in X$. 
\end{proofname}
\bibliography{refs1}

\begin{thebibliography}{10}

\bibitem{AcebronKuramoto}
J.~A. Acebron, L.~L. Bonilla, C.~J.~P. Vicente, F.~Ritort, and R.~Spigler.
\newblock The {K}uramoto model: A simple paradigm for synchronization
  phenomena.
\newblock {\em Reviews of Modern Physics}, 77(1):137--185, 2005.

\bibitem{Aey1}
D.~Aeyels and J.~A. Rogge.
\newblock Existence of partial entrainment and stability of phase locking
  behavior of coupled oscillators.
\newblock {\em Progress of Theoretical Physics}, 112(6):921--942, 2004.

\bibitem{alagoz2012user}
B.B. Alagoz, A.~Kaygusuz, and A.~Karabiber.
\newblock A user-mode distributed energy management architecture for smart grid
  applications.
\newblock {\em Energy}, 44(1):167--177, 2012.

\bibitem{anandan2017wide}
N.~Anandan and B.~George.
\newblock A wide-range capacitive sensor for linear and angular displacement
  measurement.
\newblock {\em IEEE Transactions on Industrial Electronics}, 64(7):5728--5737,
  2017.

\bibitem{aoki2015self}
T.~Aokii.
\newblock Self-organization of a recurrent network under ongoing synaptic
  plasticity.
\newblock {\em Neural Networks}, 62:11--19, 2015.

\bibitem{BacCer99}
A.~Bacciotti and F.M. Ceragioli.
\newblock Stability and stabilization of discontinuous systems and nonsmooth
  {L}yapunov functions.
\newblock {\em ESAIM: Control, Optimisation and Calculus of Variations},
  4:361--376, 1999.

\bibitem{BacCer03}
A.~Bacciotti and F.M. Ceragioli.
\newblock Nonsmooth optimal regulation and discontinuous stabilization.
\newblock {\em Abstract and Applied Analysis}, 2003(20):1159--1195, 2003.

\bibitem{bai2011cooperative}
H.~Bai, M.~Arcak, and J.~Wen.
\newblock {\em Cooperative control design: a systematic, passivity-based
  approach}.
\newblock Springer Science \& Business Media, 2011.

\bibitem{baldoni2007adaptive}
R.~Baldoni, A.~Corsaro, L.~Querzoni, S.~Scipioni, and S.~Tucci-Piergiovanni.
\newblock An adaptive coupling-based algorithm for internal clock
  synchronization of large scale dynamic systems.
\newblock In {\em OTM Confederated International Conferences ``On the Move to
  Meaningful Internet Systems''}, pages 701--716. Springer, 2007.

\bibitem{IFACKura}
R.~Bertollo, E.~Panteley, R.~Postoyan, and L.~Zaccarian.
\newblock Uniform global asymptotic synchronization of {K}uramoto oscillators
  via hybrid coupling.
\newblock In {\em IFAC World Congress}, pages 5819--5824, 2020.

\bibitem{BossoLeaderKura}
A.~Bosso, I.~A. Azzollini, S.~Baldi, and L.~Zaccarian.
\newblock Adaptive hybrid control for robust global phase synchronization of
  {K}uramoto oscillators.
\newblock {\em HAL, Also submitted for publication to the IEEE Transactions on
  Automatic Control}, hal-03372616, version 1, 2021.

\bibitem{BossoLeaderKuraCDC}
A.~Bosso, I.~A. Azzollini, S.~Baldi, and L.~Zaccarian.
\newblock A hybrid distributed strategy for robust global phase synchronization
  of second-order {K}uramoto oscillators.
\newblock In {\em IEEE Conference on Decision and Control}, pages 1212--1217,
  2021.

\bibitem{chopra}
N.~Chopra and M.~W. Spong.
\newblock On exponential synchronization of {K}uramoto oscillators.
\newblock {\em IEEE Transactions on Automatic Control}, 54(2):353--357, 2009.

\bibitem{Clarke90}
F.H. Clarke.
\newblock {\em Optimization and Nonsmooth Analysis}.
\newblock Classics in Applied Mathematics vol. 5, SIAM, 1990.

\bibitem{coraggio2020distributed}
M.~Coraggio, P.~DeLellis, and M.~di~Bernardo.
\newblock Distributed discontinuous coupling for convergence in heterogeneous
  networks.
\newblock {\em IEEE Control Systems Letters}, 5(3):1037--1042, 2020.

\bibitem{cucuzzella2018robust}
M.~Cucuzzella, S.~Trip, C.~De~Persis, X.~Cheng, A.~Ferrara, and A.~van~der
  Schaft.
\newblock A robust consensus algorithm for current sharing and voltage
  regulation in {DC} microgrids.
\newblock {\em IEEE Transactions on Control Systems Technology},
  27(4):1583--1595, 2018.

\bibitem{CuminUnsworthNeurons}
D.~Cumin and C.~P.~A. Unsworth.
\newblock Generalising the {K}uramoto model for the study of neuronal
  synchronisation in the brain.
\newblock {\em Physica D: Nonlinear Phenomena}, 226(2):181--196, 2007.

\bibitem{DePersisFrasca}
C.~{De Persis} and P.~Frasca.
\newblock Robust self-triggered coordination with ternary controllers.
\newblock {\em IEEE Transactions on Automatic Control}, 58(12):3024--3038,
  2013.

\bibitem{DellaRossaPhD21}
M.~{Della Rossa}.
\newblock {\em Non-{S}mooth {L}yapunov {F}unctions for {S}tability {A}nalysis
  of {H}ybrid {S}ystems}.
\newblock PhD Thesis, University of Toulouse, France, 2020.

\bibitem{DellaRossaMCSS21}
M.~{Della Rossa}, R.~Goebel, A.~Tanwani, and L.~Zaccarian.
\newblock Piecewise structure of {L}yapunov functions and densely checked
  decrease conditions for hybrid systems.
\newblock {\em Mathematics of Control, Signals, and Systems}, 33:123--149,
  2021.

\bibitem{DellaRossa20}
M.~{Della Rossa}, A.~Tanwani, and L.~Zaccarian.
\newblock Non-pathological {ISS-L}yapunov functions for interconnected
  differential inclusions.
\newblock {\em IEEE Transactions on Automatic Control}, 67(8):3774--3789, 2022.

\bibitem{dorfler2012synchronization}
F.~D{\"o}rfler and F.~Bullo.
\newblock Synchronization and transient stability in power networks and
  nonuniform {K}uramoto oscillators.
\newblock {\em SIAM Journal on Control and Optimization}, 50(3):1616--1642,
  2012.

\bibitem{dorfler2014synchronization}
Florian D{\"o}rfler and Francesco Bullo.
\newblock Synchronization in complex networks of phase oscillators: A survey.
\newblock {\em Automatica}, 50(6):1539--1564, 2014.

\bibitem{BulloCriticalCoupling}
F.~Dörfler and F.~Bullo.
\newblock On the critical coupling strength for {K}uramoto oscillators.
\newblock In {\em American Control Conference}, pages 3239--3244, 2011.

\bibitem{FlameDynamics}
D.~M. Forrester.
\newblock Arrays of coupled chemical oscillators.
\newblock {\em Scientific Reports}, 5(16994), 2015.

\bibitem{DesyncKura}
A.~Franci, A.~Chaillet, E.~Panteley, and F.~Lamnabhi-Lagarrigue.
\newblock Desynchronization and inhibition of {K}uramoto oscillators by scalar
  mean-field feedback.
\newblock {\em Mathematics of Control, Signals, and Systems}, 24:169--217,
  2012.

\bibitem{giraldo2019synchronisation}
J.~Giraldo, E.~Mojica-Nava, and N.~Quijano.
\newblock Synchronisation of heterogeneous {K}uramoto oscillators with sampled
  information and a constant leader.
\newblock {\em International Journal of Control}, 92(11):2591--2607, 2019.

\bibitem{GodsilGraphTheory}
C.~Godsil and G.~Royle.
\newblock {\em Algebraic Graph Theory}.
\newblock Springer, 2001.

\bibitem{TeelBook12}
R.~Goebel, R.~G. Sanfelice, and A.~R. Teel.
\newblock {\em Hybrid Dynamical Systems: modeling, stability, and robustness}.
\newblock Princeton University Press, 2012.

\bibitem{hajek1979discontinuous}
Otomar H{\'a}jek.
\newblock Discontinuous differential equations, {I}.
\newblock {\em Journal of Differential Equations}, 32(2):149--170, 1979.

\bibitem{jad1}
A.~Jadbabaie, N.~Motee, and M.~Barahona.
\newblock On the stability of the {K}uramoto model of coupled nonlinear
  oscillators.
\newblock In {\em American Control Conference}, pages 4296--4301, 2004.

\bibitem{JafarpourBullo}
S.~Jafarpour and F.~Bullo.
\newblock Synchronization of {K}uramoto oscillators via cutset projections.
\newblock {\em IEEE Transactions on Automatic Control}, 64(7):2830--2844, 2019.

\bibitem{kiss2018synchronization}
Istv{\'a}n~Z. Kiss.
\newblock Synchronization engineering.
\newblock {\em Current Opinion in Chemical Engineering}, 21:1--9, 2018.

\bibitem{kuramoto1975self}
Y.~Kuramoto.
\newblock Self-entrainment of a population of coupled non-linear oscillators.
\newblock In {\em International Symposium on Mathematical Problems in
  Theoretical Physics}, pages 420--422. Springer, 1975.

\bibitem{LeonardAnimalGroups}
N.~E. Leonard, T.~Shen, B.~Nabet, L.~Scardovi, I.~D. Couzin, and S.~A. Levin.
\newblock Decision versus compromise for animal groups in motion.
\newblock {\em Proceedings of the National Academy of Sciences},
  109(1):227--232, 2012.

\bibitem{mauroy2012contraction}
Alexandre Mauroy and Rodolphe Sepulchre.
\newblock Contraction of monotone phase-coupled oscillators.
\newblock {\em Systems \& Control Letters}, 61(11):1097--1102, 2012.

\bibitem{HTsync}
C.~G. Mayhew, M.~Arcak R.~G.~Sanfelice, J.~Sheng, and A.~R. Teel.
\newblock Quaternion-based hybrid feedback for robust global attitude
  synchronization.
\newblock {\em IEEE Transactions on Automatic Control}, 57(8):2122--2127, 2012.

\bibitem{mayhew2012path}
C.G. Mayhew, R.G. Sanfelice, and A.R. Teel.
\newblock On path-lifting mechanisms and unwinding in quaternion-based attitude
  control.
\newblock {\em IEEE Transactions on Automatic Control}, 58(5):1179--1191, 2012.

\bibitem{mesbahi2010graph}
M.~Mesbah and M.~Egerstedt.
\newblock {\em Graph Theoretic Methods in Multiagent Networks}.
\newblock Princeton University Press, 2010.

\bibitem{miller1997maneuvering}
R.~B. Miller and M.~Pachter.
\newblock Maneuvering flight control with actuator constraints.
\newblock {\em Journal of Guidance, Control, and Dynamics}, 20(4):729--734,
  1997.

\bibitem{OudMetronomes}
W.~T. Oud.
\newblock {\em Design and experimental results of synchronizing metronomes,
  inspired by {C}hristiaan {H}uygens}.
\newblock Master's Thesis, Eindhoven University of Technology, 2006.

\bibitem{paley2007oscillator}
Derek~A Paley, Naomi~Ehrich Leonard, Rodolphe Sepulchre, Daniel Grunbaum, and
  Julia~K Parrish.
\newblock Oscillator models and collective motion.
\newblock {\em IEEE Control Systems Magazine}, 27(4):89--105, 2007.

\bibitem{pandurangan2018distributed}
G.~Pandurangan, P.~Robinson, and M.~Scquizzato.
\newblock The distributed minimum spanning tree problem.
\newblock {\em Bulletin of EATCS}, 2(125), 2018.

\bibitem{Polyakov11}
A.~Polyakov.
\newblock Nonlinear feedback design for fixed-time stabilization of linear
  control systems.
\newblock {\em IEEE Transactions on Automatic Control}, 57(8):2106--2110, 2011.

\bibitem{rad2011lower}
A.A. Rad, M.~Jalili, and M.~Hasler.
\newblock A lower bound for algebraic connectivity based on the
  connection-graph-stability method.
\newblock {\em Linear algebra and its applications}, 435(1):186--192, 2011.

\bibitem{reigosa2018permanent}
D.~Reigosa, D.~Fernandez, C.~Gonzalez, S.B. Lee, and F.~Briz.
\newblock Permanent magnet synchronous machine drive control using analog
  hall-effect sensors.
\newblock {\em IEEE Transactions on Industry Applications}, 54(3):2358--2369,
  2018.

\bibitem{SanfeliceToolbox}
R.~G. Sanfelice, D.~Copp, and P.~Nanez.
\newblock A toolbox for simulation of hybrid systems in {M}atlab/{S}imulink:
  {H}ybrid {E}quations ({H}y{EQ}) {T}oolbox.
\newblock In {\em Proceedings of the 16th International Conference on Hybrid
  Systems: Computation and Control}, pages 101--106. ACM, 2013.

\bibitem{Scholtes}
S.~Scholtes.
\newblock {\em Introduction to Piecewise Differentiable Equations}.
\newblock SpringerBriefs in Optimization, Springer, 2012.

\bibitem{sepulchre2007stabilization}
R.~Sepulchre, D.~A. Paley, and N.~E. Leonard.
\newblock Stabilization of planar collective motion: All-to-all communication.
\newblock {\em IEEE Transactions on Automatic Control}, 52(5):811--824, 2007.

\bibitem{Krstic17}
Y.~Song, Y.~Wang, J.~Holloway, and M.~Krsti{ć}.
\newblock Time-varying feedback for regulation of normal-form nonlinear systems
  in prescribed finite time.
\newblock {\em Automatica}, 83:243--251, 2017.

\bibitem{SontagIOS}
E.~D. Sontag and Y.~Wang.
\newblock On characterizations of the input-to-state stability property.
\newblock {\em Systems and Control Letters}, 24(1):351--359, 1995.

\bibitem{Strogatz00}
S.~H. Strogatz.
\newblock From {K}uramoto to {C}rawford: Exploring the onset of synchronization
  in populations of coupled oscillators.
\newblock {\em Physica D: Nonlinear Phenomena}, 143(1):1--20, 2000.

\bibitem{strogatz2003sync}
S.H. Strogatz.
\newblock {\em Sync: The Emerging Science of Spontaneous Order}.
\newblock Hyperion, NY, 2003.

\bibitem{TassDeepBrain}
P.~A. Tass.
\newblock A model of desynchronizing deep brain stimulation with a
  demand-controlled coordinated reset of neural subpopulations.
\newblock {\em Biological Cybernetics}, 89(2):81--88, 2003.

\bibitem{Valadier89}
M.~Valadier.
\newblock Entra{î}nement unilat{é}ral, lignes de descente, fonctions
  {L}ipschitziennes non pathologiques.
\newblock {\em CRAS Paris}, 308:241--244, 1989.

\bibitem{Wu18}
J.~Wu and X.~Li.
\newblock Finite-time and fixed-time synchronization of {K}uramoto-oscillator
  network with multiplex control.
\newblock {\em IEEE Transactions on Control of Network Systems}, 6(2):863--873,
  2018.

\bibitem{Xu16}
J.~Xu, T.~J.~J. van~den Boom, B.~De Schutter, and S.~Wang.
\newblock Irredundant lattice representations of continuous piecewise affine
  functions.
\newblock {\em Automatica}, 70:109--120, 2017.

\end{thebibliography}
\bibliographystyle{plain}
\end{document}